\documentclass[english,11pt,a4paper]{article}
\usepackage[T1]{fontenc}
\usepackage[latin9]{inputenc}
\usepackage{geometry}
\usepackage{babel}
\usepackage{hyperref}
\usepackage{enumitem}

\usepackage{amsmath, amssymb, amstext, mathtools}
\usepackage{amsthm}

\usepackage{tikz}
\usepackage{pgfplots}
\pgfplotsset{compat=1.18}
\usepackage{textpos}

\definecolor[named]{urlblue}{cmyk}{1,0.58,0,0.21}

\hypersetup{
    breaklinks=true,
    colorlinks=true,
    citecolor=purple!70!blue!60!black,
    linkcolor=purple!70!blue!90!black,
    urlcolor=urlblue,
    pdflang={en},
    pdftitle={Counting Small Induced Subgraphs: Hardness via Fourier Analysis},
    pdfauthor={Radu Curticapean and Daniel Neuen}
}

\usetikzlibrary{backgrounds,patterns,arrows,decorations.pathreplacing,decorations.pathmorphing,calc}

\tikzstyle{vertex}=[draw,circle,fill=white,minimum size=7pt,inner sep=0pt]

\newtheorem{theorem}{Theorem}[section]
\newtheorem{lemma}[theorem]{Lemma}
\newtheorem{corollary}[theorem]{Corollary}

\newtheorem{conjecture}[theorem]{Conjecture}
\newtheorem{claim}[theorem]{Claim}
\newtheorem{fact}[theorem]{Fact}
\theoremstyle{definition}
\newtheorem{definition}[theorem]{Definition}

\theoremstyle{remark}
\newtheorem{remark}[theorem]{Remark}

\newenvironment{claimproof}{\begin{proof}}{\end{proof}}

\newcommand{\CA}{{\mathcal A}}
\newcommand{\CB}{{\mathcal B}}
\newcommand{\CC}{{\mathcal C}}

\newcommand{\CL}{{\mathcal L}}

\newcommand{\NN}{{\mathbb N}}
\newcommand{\ZZ}{{\mathbb Z}}
\newcommand{\RR}{{\mathbb R}}
\newcommand{\QQ}{{\mathbb Q}}
\newcommand{\FF}{{\mathbb F}}

\newcommand{\ba}{\boldsymbol{a}}
\newcommand{\bb}{\boldsymbol{b}}
\newcommand{\bc}{\boldsymbol{c}}
\newcommand{\bx}{\boldsymbol{x}}
\newcommand{\by}{\boldsymbol{y}}
\newcommand{\bz}{\boldsymbol{z}}

\newcommand{\im}{\mathrm{im}}

\newcommand{\numindsub}[2]{\#\mathrm{IndSub}(#1 \to #2)}
\newcommand{\numindsubstar}[1]{\#\mathrm{IndSub}(#1 \to \,\star\,)}

\newcommand{\numsub}[2]{\#\mathrm{Sub}(#1 \to #2)}
\newcommand{\numsubstar}[1]{\#\mathrm{Sub}(#1 \to \,\star\,)}

\newcommand{\subpoly}[1]{\mathrm{Sub}_{#1}}
\newcommand{\indsubpoly}[1]{\mathrm{IndSub}_{#1}}

\newcommand{\numhom}[2]{\#\mathrm{Hom}(#1 \to #2)}

\newcommand{\pclique}{\#\textnormal{\textsc{Clique}}}
\newcommand{\pmodclique}[1]{\#_{#1}\textnormal{\textsc{Clique}}}

\newcommand{\psub}[1]{\#\textnormal{\textsc{ParSub}}(#1)}

\newcommand{\pcolsub}[1]{\#\textnormal{\textsc{ParColSub}}(#1)}
\newcommand{\pmodcolsub}[2]{\#_{#1}\textnormal{\textsc{ParColSub}}(#2)}
\newcommand{\colsub}[1]{\#\textnormal{\textsc{ColSub}}(#1)}
\newcommand{\modcolsub}[2]{\#_{#1}\textnormal{\textsc{ColSub}}(#2)}
\newcommand{\deccolsub}[1]{\textnormal{\textsc{ColSub}}(#1)}

\newcommand{\pindsub}[1]{\#\textnormal{\textsc{ParIndSub}}(#1)}
\newcommand{\pmodindsub}[2]{\#_{#1}\textnormal{\textsc{ParIndSub}}(#2)}
\newcommand{\sindsub}[1]{\#\textnormal{\textsc{IndSub}}(#1)}
\newcommand{\smodindsub}[2]{\#_{#1}\textnormal{\textsc{IndSub}}(#2)}
\newcommand{\kindsub}[2]{\#\textnormal{\textsc{IndSub}}(#1,#2)}

\newcommand{\slice}[2]{#1^{(#2)}}
\newcommand{\altenum}[1]{#1^{\downarrow}}

\newcommand{\hw}{{\sf hw}}

\newcommand{\True}{{\sf True}}
\newcommand{\False}{{\sf False}}

\newcommand{\alleven}{\Phi_{\sf even}}
\newcommand{\altalleven}{\altenum{\Phi_{\sf even}}}

\newcommand{\can}[1]{#1^{\textnormal{\textsf{can}}}}

\DeclareMathOperator{\tw}{tw}
\DeclareMathOperator{\vc}{vc}

\DeclareMathOperator{\supp}{supp}
\DeclareMathOperator{\Aut}{Aut}
\DeclareMathOperator{\Sym}{Sym}

\newcommand{\sharpwone}{\textnormal{\textsf{\#W[1]}}}
\newcommand{\modwone}[1]{\textnormal{\textsf{Mod}}_{#1}\textnormal{\textsf{W[1]}}}
\newcommand{\sharpP}{\textnormal{\textsf{\#P}}}

\newcommand{\orcid}[1]{\href{https://orcid.org/#1}{\includegraphics[height=1.8ex]{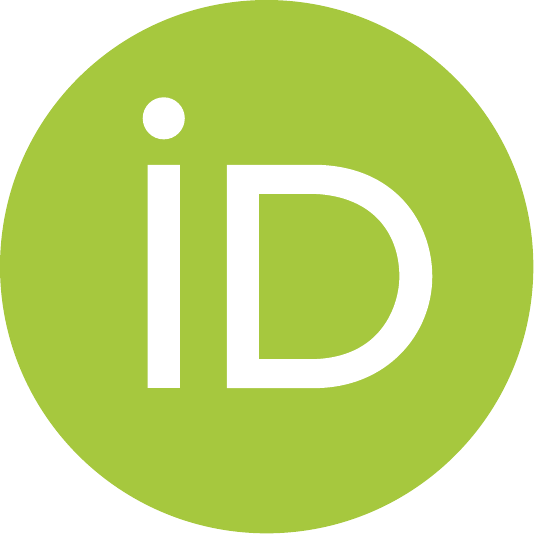}}}

\title{Counting Small Induced Subgraphs:\\ Hardness via Fourier Analysis%
\thanks{\rightskip=5cm 
Funded by the European Union (ERC, CountHom, 101077083). Views and opinions expressed are however those of the author(s) only and do not necessarily reflect those of the European Union or the European Research Council Executive Agency. Neither the European Union nor the granting authority can be held responsible for them.}
}
\author{Radu Curticapean \orcid{0000-0001-7201-9905}\\
University of Regensburg and IT University of Copenhagen
\and
Daniel Neuen \orcid{0000-0002-4940-0318}\\
University of Regensburg and Max Planck Institute for Informatics}

\date{}

\begin{document}

\maketitle
\begin{textblock}{5}(8.5, 8.75) \includegraphics[width=120px]{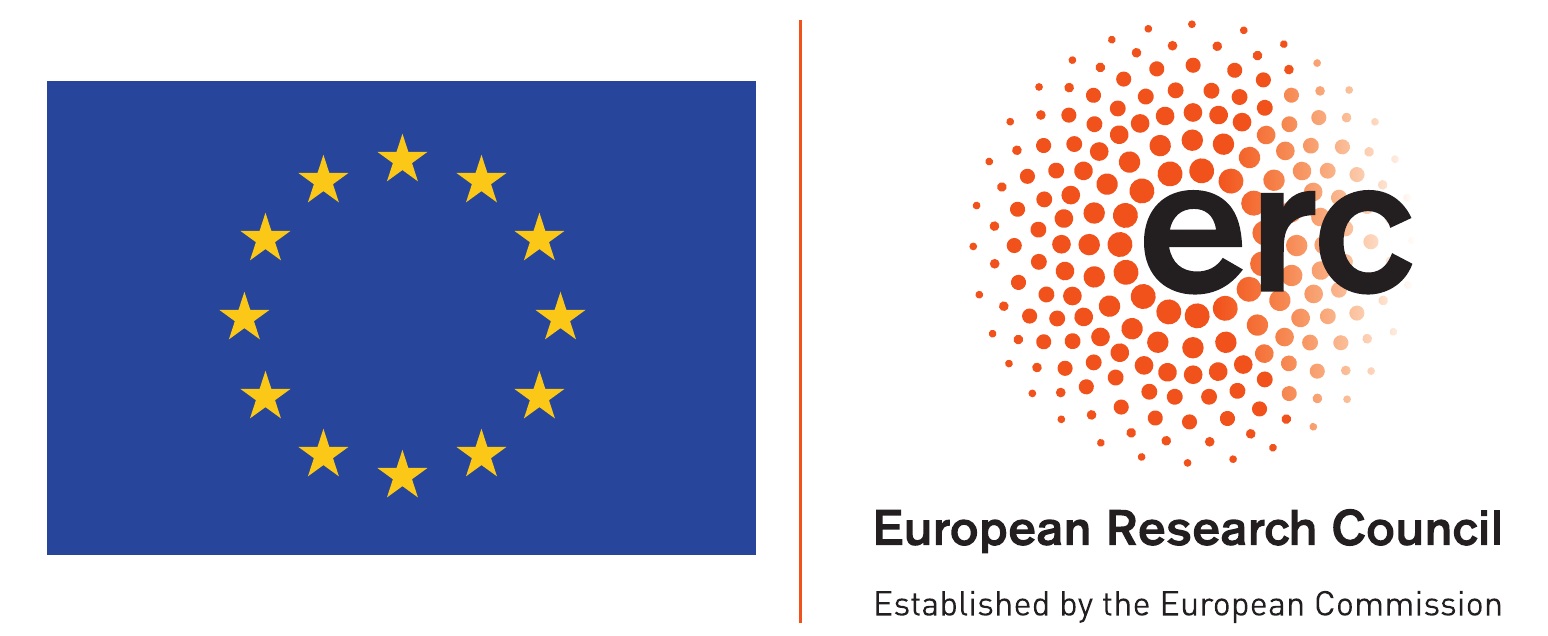} \end{textblock}

\begin{abstract}
    For a fixed graph property $\Phi$ and integer $k \geq 1$, consider the problem of counting the induced $k$-vertex subgraphs satisfying $\Phi$ in an input graph $G$.
    This problem can be solved by brute-force in time $O(n^{k})$.
    Under ETH, we prove several lower bounds on the optimal exponent in this running time:
    \begin{itemize}
        \item If $\Phi$ is edge-monotone (i.e.,  closed under deleting edges), then ETH rules out $n^{o(k)}$ time algorithms for this problem. This strengthens a recent lower bound by D{\"{o}}ring, Marx and Wellnitz~[STOC 2024].
        Our result also holds for counting modulo fixed primes.
        \item If at most $(2-\varepsilon)^{\binom{k}{2}}$ graphs on $k$ vertices satisfy $\Phi$, for some $\varepsilon > 0$, then ETH also rules out an exponent of $o(k)$.
        This holds even when the graphs in $\Phi$ have arbitrary individual weights, generalizing previous results for hereditary properties by Focke and Roth~[SIAM J.\ Comput.\ 2024].
        \item If $\Phi$ is non-trivial and excludes $\beta_\Phi$ edge-densities, then the optimal exponent under ETH is $\Omega(\beta_\Phi)$.
        This holds even when the graphs in $\Phi$ have arbitrary individual weights, generalizing previous results by Roth, Schmitt and Wellnitz~[SIAM J.\ Comput.\ 2024].
    \end{itemize}
    In all cases, we also obtain $\sharpwone$-hardness if $k$ is part of the input and considered as the parameter.
    We also obtain lower bounds on the Weisfeiler-Leman dimension.

    As opposed to the nontrivial techniques from combinatorics, group theory, and simplicial topology used before,
    our results follow from a relatively straightforward ``algebraization'' of the problem in terms of polynomials, combined with applications of simple algebraic facts, which can also be interpreted in terms of Fourier analysis.
    Most of the $\sharpwone$-hardness results known in the area are subsumed by our paper, and our hardness results often hold in a more general setting.
\end{abstract}

\section{Introduction}
Counting small patterns in graphs is of fundamental importance in computer science, with practical applications in network analysis and bioinformatics~\cite{MiloSIKCA02}.
From a theoretical perspective, pattern counting problems can be formalized in the following ways:

\begin{enumerate}
    \item For a \emph{fixed pattern} $H$, consider the function $\numindsubstar{H}$ that, on input a graph $G$, outputs the number $\numindsub{H}{G}$ of induced $H$-copies in $G$. An example is the number of triangles in $G$.
    Each such function can be computed in polynomial time, and we wish to use hardness assumptions such as the Exponential-Time Hypothesis (ETH)~\cite{ImpagliazzoPZ01} to bound the exponent of this running time, depending on $H$.
    \item We may also view both $H$ and $G$ as input when computing $\numindsub{H}{G}$ and then consider $k \coloneqq |V(H)|$ as a parameter, in the sense of \emph{parameterized complexity}~\cite{CyganFKLMPPS15,FlumG06}.
    Here, we seek fixed-parameter tractable algorithms with running time $f(k) \cdot n^{O(1)}$ for computable functions $f$.
    Similar to the classical notion of $\sharpP$-hardness, the notion of $\sharpwone$-hardness in parameterized complexity rules out such algorithms.
\end{enumerate}

In the fixed-pattern perspective, it is known that ETH implies a constant $\alpha > 0$ such that, for every $k$-vertex graph $H$, the number  $\numindsubstar{H}$ cannot be computed in time $O(n^{\alpha \cdot k})$~\cite{ChenCFHJKX05,DoringMW24}.
Note that this lower bound holds for every pattern $H$.

From the parameterized perspective, the problem of counting induced $H$-copies in graphs $G$ with parameter $|V(H)|$ is $\sharpwone$-hard, since it subsumes the canonical $\sharpwone$-hard problem of counting $k$-cliques~\cite{FlumG06}.
For a more refined picture, we consider the problems $\pindsub{\mathcal H}$ for fixed recursively enumerable graph classes $\mathcal H$:
Here, the input is a graph $H\in \mathcal H$ and a graph $G$, the output is $\numindsub{H}{G}$, and the parameter is $|V(H)|$. Every such problem is trivially polynomial-time solvable if $\mathcal H$ is finite; otherwise it is $\sharpwone$-complete~\cite{ChenTW08}.

\subsection{Counting Induced Subgraphs Satisfying a Fixed Property}

Rather than counting induced $H$-copies for fixed $H$, one may be interested in fixing a \emph{set} $\Phi$ of $k$-vertex graphs and counting induced subgraphs of $G$ that are isomorphic to some graph $H \in \Phi$.
We denote this value by $\numindsub{\Phi}{G}$.
The systematic study of such counting problems was initiated by Jerrum and Meeks \cite{JerrumM15,JerrumM15b,JerrumM17}:
Given a fixed graph property $\Phi$, they define the parameterized problem $\pindsub{\Phi}$ that asks, on input $(G,k)$, to determine the number $\numindsub{\slice{\Phi}{k}}{G}$ of induced $k$-vertex subgraphs of $G$ that satisfy $\Phi$.
Here, the \emph{slice} $\slice{\Phi}{k}$ denotes the restriction of $\Phi$ to graphs on exactly $k$ vertices, while $\Phi$ itself may contain graphs of arbitrary sizes.

It has been shown that $\pindsub{\Phi}$ is either fixed-parameter tractable or $\sharpwone$-hard, for every computable property $\Phi$~\cite{CurticapeanDM17}, with an explicit tractability criterion.
However, this criterion is indirect, meaning it is often highly nontrivial to determine whether a given property $\Phi$ meets it.
In particular, the criterion has not helped in identifying nontrivial tractable properties or ruling out their existence.
To elaborate, if all but finitely many slices $\slice{\Phi}{k}$ are trivial (we also say $k$-trivial) in the sense that they are empty or contain all $k$-vertex graphs, then $\pindsub{\Phi}$ is easy to compute.
Such properties $\Phi$ are called \emph{meager}; for all other properties, it was conjectured that $\pindsub{\Phi}$ is $\sharpwone$-hard~\cite{DorflerRSW22,FockeR24,RothSW24}; see also \cite{Roth26} for a recent survey.

It was shown very recently that this conjecture does not hold in the full generality stated \cite{CurticapeanDN25}.
Nevertheless, the conjecture has been verified for several very natural classes of properties; see~\cite[Section~1.1]{DoringMW24}.
We list some of these known results; whenever we say in the following that ``ETH implies an exponent of $\Omega(f(k))$'', we mean that ETH implies a constant $\alpha > 0$ such that $\numindsub{\slice{\Phi}{k}}{G}$ cannot be computed in time $O(n^{\alpha \cdot f(k)})$ on $n$-vertex graphs, for an infinite set of integers $k$. In the following, we assume $\Phi$ to be computable.

\begin{enumerate}
    \item If $\Phi$ is closed under deleting edges (i.e., edge-monotone) \cite{DoringMW24} or closed under deleting vertices (i.e., hereditary) \cite{FockeR24}, then $\pindsub{\Phi}$ is $\sharpwone$-hard. This strengthens an earlier result that required $\Phi$ to satisfy both closure properties \cite{RothSW24}.
    It has also been shown that ETH implies an exponent of $\Omega(\sqrt{\log k} / \log \log k)$ for edge-monotone properties, and an exponent of $\Omega(k)$ for hereditary properties.
    \item Given $k\in \NN$, let $\beta_\Phi(k)$ be the number of edge-counts avoided by $\Phi$ on $k$-vertex graphs.
    If there is an infinite subset of $\NN$ on which $\Phi$ is slice-wise nontrivial and $\beta_\Phi$ grows like $\omega(k)$, then $\pindsub{\Phi}$ is $\sharpwone$-hard~\cite{RothSW24}.
    Moreover, ETH implies an exponent of $\Omega(\beta_\Phi(k)/\sqrt{\log \beta_\Phi(k)})$.
    \item If $\Phi(H)$ depends only on $|E(H)|$, then the problem $\pindsub{\Phi}$ is $\sharpwone$-hard and ETH implies an exponent of $\Omega(k/\sqrt{\log k})$ \cite{RothSW24}.
    \item If there are infinitely many prime powers $t$ such that $\Phi$ holds for exactly one of the graphs $K_{t,t}$ (the complete bipartite graph) or $I_{2t}$ (the edge-less graph),
    then $\pindsub{\Phi}$ is $\sharpwone$-hard and ETH implies an exponent of $\Omega(k)$~\cite{DorflerRSW22}.
\end{enumerate}

In this paper, we strengthen most of the above hardness results.
However, we consider our main contribution to be the introduction of simple algebraic techniques to the study of $\pindsub{\Phi}$, which was previously mostly approached through sophisticated techniques from group theory. Indeed, we derive many of the above results as relatively straightforward consequences of known results such as the Schwartz-Zippel lemma.

\subsection{The Known Technique: Using the Sub-Expansion}

Following other recent works, we approach $\pindsub{\Phi}$ by interpreting its $k$-vertex slices $\slice{\Phi}{k}$ as linear combinations of (not necessarily induced) \emph{subgraph} counts.
Let us fix $k \in \NN$ and abbreviate $f(\star) = \numindsubstar{\slice{\Phi}{k}}$.
Then
\begin{equation}
    \label{eq:intro-indsub-lincomb}
    f(\star) = \underbrace{\sum_{H \in \slice{\Phi}{k}} \alpha_H \cdot \numindsubstar{H}}_{\text{ind-expansion of } f(\star)},
\end{equation}
with $H$ ranging over unlabeled $k$-vertex graphs, where we have $\alpha_H = 1$ for all $H \in \slice{\Phi}{k}$ and $\alpha_H = 0$ otherwise.
Note that $\slice{\Phi}{k}$ is finite; we could generally allow linear combinations involving arbitrary finite sets of graphs.

By known formulas~\cite[(5.17)]{Lovasz12}, there exists a set $\mathcal H$ of unlabeled $k$-vertex graphs, together with non-zero coefficients $\beta_H \in \ZZ$ for $H \in \mathcal H$, such that
\begin{equation}
    \label{eq:intro-sub-lincomb}
    f(\star) = \underbrace{\sum_{H \in \mathcal H}\beta_{H} \cdot \numsubstar{H}}_{\text{sub-expansion of } f(\star)}.
\end{equation}

We call the linear combinations on the right-hand side of \eqref{eq:intro-indsub-lincomb} and \eqref{eq:intro-sub-lincomb} the \emph{ind-expansion} and \emph{sub-expansion} of $f(\star)$, respectively.
These expansions are specified by the involved patterns and their (non-zero) coefficients.
Both expansions are unique for a given function $f$, as the involved functions $\numindsubstar{H}$ and $\numsubstar{H}$ are linearly independent for simple unlabeled graphs, see, e.g., \cite[Corollary~5.45]{Lovasz12}.

\paragraph*{Complexity Monotonicity.}
Linear combinations of $k$-vertex subgraph pattern counts as above enjoy a useful ``complexity monotonicity'': Under certain conditions, the presence of a single hard function $\numsubstar{H}$ in the linear combination can render the entire linear combination hard.
This phenomenon was already observed and exploited for \emph{homomorphism} counts~\cite{CurticapeanDM17}, for which it holds without further conditions.
For subgraph counts, a sufficient condition for such a ``complexity monotonicity'' is that all patterns have the same number $k$ of vertices, and moreover not only $\numsubstar{H}$ is hard, but even its \emph{colorful} variant is hard.
(In this variant, given a vertex-colored graph $G$, we count subgraph copies of $H$ that are colorful, i.e., which contain exactly one vertex of each color in $G$.
The colorful variant admits a parameterized reduction to the uncolored variant, but the converse does not hold, e.g., for matchings.)

We can thus understand the complexity of counting induced subgraphs with property $\Phi$ by answering two questions: First, for which $H$ is counting colorful $H$-subgraphs hard? Second, which patterns $H$ appear in the sub-expansion \eqref{eq:intro-sub-lincomb} derived from $\Phi$?
Similar approaches were taken in~\cite{DorflerRSW22,DoringMW24,RothSW24}.
We discuss these two questions in the following.

\paragraph{Which Subgraphs are Hard to Count?}
The complexity of counting subgraphs is well-understood in both the uncolored and colorful settings~\cite{CurticapeanM14,CurticapeanDM17}.
First, let us observe that counting $H$-subgraphs may be easier than counting induced $H$-subgraphs, e.g., when $H$ is edgeless.

In the \emph{uncolored} setting, it is known that the subgraph variant $\psub{\mathcal H}$ of the induced subgraph counting problem $\pindsub{\mathcal H}$ is polynomial-time solvable if the maximum vertex-cover number among graphs in $\mathcal H$ is finite;
otherwise it is $\sharpwone$-hard and almost tight lower bounds are known under ETH~\cite{CurticapeanM14,CurticapeanDM17}.

The \emph{colorful} variant admits additional tractable patterns: Standard techniques yield an $O(n^{\tw(H)+1})$ time algorithm for counting colorful $H$-subgraphs in vertex-colored $n$-vertex graphs, where $\tw(H)$ is the \emph{treewidth} of $H$, which measures how similar $H$ is to a tree for algorithmic purposes.
For recursively enumerable $\mathcal H$ of bounded treewidth,
the problem $\pcolsub{\mathcal H}$ of counting colorful $H$-subgraphs in $G$ is therefore polynomial-time solvable, where both $H\in \mathcal H$ and $G$ are inputs.
If the treewidth of $\mathcal H$ is unbounded, then the problem is $\sharpwone$-hard, see~\cite[Theorem~3.4]{CurticapeanM14} or~\cite[Theorem~5.6]{Curticapean15}.
Moreover, ETH implies a constant $\beta > 0$ such that for every graph $H$ with treewidth $t$, the number of colorful $H$-subgraph copies in $n$-vertex graphs cannot be computed in time $O(n^{\beta \cdot t / \log t})$~\cite{Marx10}.

In this paper, we rely on large edge-density of $H$ as a sufficient proxy condition to guarantee large treewidth of $H$ and thus hardness of counting colorful $H$-subgraph copies.
Indeed, it is well-known that graphs $H$ of treewidth $t$ have at most $t \cdot |V(H)|$ edges.
Moreover, the following lower bound under ETH was recently shown.

\begin{theorem}[{\cite[Theorem 1.3]{CurticapeanDNW25}}]
    \label{fact:density-makes-hard}
    Assuming ETH, there is a constant $\beta > 0$ such that, for every pattern $H$ with $k$ vertices and $\ell$ edges, colorful $H$-subgraph copies in $n$-vertex cannot be counted in $O(n^{\beta \cdot \ell/k})$ time.
\end{theorem}

A similar version of this theorem holds with regard to $\sharpwone$-hardness.
We stress that large edge-density is only a \emph{sufficient} criterion for hardness.
Indeed, expanders provide an example of patterns $H$ that have degree $3$ and treewidth $\Omega(|V(H)|)$.

\paragraph{Which Patterns Appear in the Sub-Expansion?}
Theorem~\ref{fact:density-makes-hard} with the complexity monotonicity of subgraph counts allow us to show hardness of $f(\star)=\numindsubstar{\slice{\Phi}{k}}$ by finding $k$-vertex patterns $H$ with $\omega(k)$ edges and non-zero coefficients in the sub-expansion~\eqref{eq:intro-sub-lincomb} of $f(\star)$.

As a first step, it can be shown that the coefficient $\beta_H$ of $H$ in the sub-expansion equals a certain combinatorial quantity, the \emph{alternating enumerator}.
We define this quantity as
\[
\altenum{\Phi}(H) = (-1)^{|E(H)|} \sum_{S \subseteq E(H)} (-1)^{|S|} \cdot \Phi(H[S]),
\]
where $H[S]$ is the subgraph obtained from $H$ by keeping all vertices and only the edges in $S$.
Our definition includes an extra factor of $(-1)^{|E(H)|}$ compared to previous works~\cite{DorflerRSW22,Roth19,RothS20}.

However, showing that $\altenum{\Phi}(H) \neq 0$ turns out to be nontrivial:
With some exceptions~\cite{FockeR24,RothSW24}, most approaches towards this task this can be traced back to a seminal paper~\cite{RothS20} that interprets edge-monotone properties $\Phi$ defined on $k$-vertex graphs as abstract simplicial complexes.
Up to sign, $\altenum{\Phi}(K_k)$ then is the (reduced) \emph{Euler characteristic} of the simplicial complex: Viewing $t$-edge graphs as dimension-$t$ facets of the complex corresponding to $\Phi$, we obtain that $\altenum{\Phi}(K_k)$ is the difference between the counts of odd-dimension and even-dimension facets, similar to Euler's formula for plane graphs.
By analyzing group actions on simplicial complexes that preserve the Euler characteristic modulo primes, it is then shown that $\altenum{\Phi}(K_k) \neq 0$ for several properties $\Phi$, e.g., for edge-monotone properties $\Phi$ that are false on all odd cycles.
In subsequent works~\cite{DorflerRSW22,DoringMW24,Roth0W21}, transitive group actions on graphs are used to derive more general results, also for related pattern counting problems, culminating in the recent result on edge-monotone properties~\cite{DoringMW24} mentioned above.

\subsection{Our Approach: Understanding Sub-Expansions via Polynomials}

To summarize the method described in the previous section, we have a graph parameter
\[
f(\star) = \sum_{H} \alpha_H \cdot \numindsubstar{H} = \sum_{H} \beta_H \cdot \numsubstar{H}
\] 
involving $k$-vertex graphs $H$ in the sums. Both representations are unique, and the coefficients $\beta_H$ are given by the $\alpha_H$ via alternating enumerators.
As outlined above, $f(\star)$ is hard if there is a graph $H$ with $\beta_H \neq 0$ and $\omega(k)$ edges.
Our goal is to find such dense graphs $H$.

To attain this goal, we propose the following new strategy: First, we generalize graph parameters $f(\star)$ of the above form to multilinear polynomials.
Related polynomials were already studied before in different contexts in algebraic complexity theory~\cite{BhargavCC025,DawarPS25,KomarathPR23}.
Then, we observe that the degree of the polynomial corresponding to $f$ is precisely the maximum edge-count $|E(H)|$ among graphs $H$ with $\beta_H \neq 0$.
Finally, we establish degree lower bounds on the polynomials using various algebraic techniques.

\paragraph{Polynomials for Graph Invariants.}
 Let $t \in \mathbb N$, let $K_t$ be the complete graph on $t$ vertices, and let $x_{uv}$ for $uv \in E(K_t)$ be indeterminates.
We define the \emph{subgraph polynomial}
\begin{equation}
    \label{eq:sub-polynomial}
    \subpoly{H,t} := \sum_{\substack{F \subseteq K_t \\F \cong H}} \, \prod_{uv \in E(F)} x_{uv},
\end{equation}
where $F$ ranges over all subgraphs of $K_t$ isomorphic to $H$.
This can be interpreted as a weighted $H$-subgraph count in $K_t$.
We also define the \emph{induced subgraph polynomial}
\begin{equation}
    \label{eq:indsub-polynomial}
    \indsubpoly{H,t} := \sum_{\substack{F \subseteq K_t \\F \cong H}} \, \prod_{uv \in E(F)} x_{uv} \prod_{uv \in \binom{V(F)}{2} \setminus E(F)} (1-x_{uv}).
\end{equation}

These polynomials generalize subgraph and induced subgraph counts: 
Given a graph $G$ with $V(G) = [t]$ for $t \in \mathbb N$, and writing $m = \binom{t}{2}$, let $\ba_G \in \{0,1\}^m$ be the characteristic vector of $E(G)$ as a subset of $\binom{[t]}{2}$. Then we have
\begin{align}
\label{eq:intro-fourier-subexpansion}
    \subpoly{H,t}(\ba_G) & = \numsub{H}{G}, \\
    \indsubpoly{H,t}(\ba_G) & = \numindsub{H}{G}.
\end{align}

Since multilinear polynomials are determined by their evaluations on $0$-$1$-inputs, the  polynomials $\subpoly{H,t}$ and $\indsubpoly{H,t}$ are the unique multilinear polynomials that agree, on $t$-vertex graphs, with subgraph counts from $H$ or induced subgraph counts from $H$.
In other words, not only do these  polynomials allow us to recover (induced) subgraph counts on $t$-vertex graphs at $0$-$1$-inputs, the polynomials are even \emph{determined} by these counts.

This holds more generally for linear combinations of such counts:
As in \eqref{eq:intro-sub-lincomb}, let $f(\star) = \sum_F \beta_F \cdot \numsubstar{F}$ be a linear combination of subgraph counts.
Then the linear combination of polynomials $q_{f,t} = \sum_F \beta_F \cdot \subpoly{F,t}$ is the unique polynomial that agrees with $f$ on $t$-vertex graphs. We call $q_{f,t}$ the \emph{polynomial representation}
of $f$ on $t$-vertex graphs.

\paragraph{Large Degree Implies Hardness.}
For all graphs $H$ and integers $t \geq |V(H)|$, the degree of the subgraph  polynomial $\subpoly{H,t}$ is $|E(H)|$.
This follows immediately from the monomial expansion \eqref{eq:sub-polynomial}.
More generally, for every linear combination $f(\star) = \sum_F \beta_F \cdot \numsubstar{F}$, the degree of the  polynomial $q_{f,t} = \sum_F \beta_F \cdot \subpoly{F,t}$ is at most the maximum of $|E(F)|$ over all graphs $F$ with $\beta_F \neq 0$. (In fact, provided that $t \geq k$, the degree is even \emph{equal} to this maximum. This follows from the linear independence of the polynomials, but it is not required in our arguments.)

This simple fact turns out to be very useful:
Given a graph parameter $f(\star)$ that is a finite linear combination of subgraph counts, if the degree of its  polynomial representation
is at least $d$ for some $t \in \NN$, then the sub-expansion of $f(\star)$ must contain a graph $F$ with at least $d$ edges.
In particular, if $d\in \omega(k)$, then the strategy described in the beginning of this subsection comes to fruition and we obtain the desired hardness results.

As an example, reasoning along this strategy shows that $f(\star) = \numindsubstar{H}$ is hard for \emph{every} $k$-vertex graph $H$, since its unique representation as a multilinear polynomial \eqref{eq:indsub-polynomial} clearly has degree $\binom{k}{2}$. This means that the sub-expansion of $\numindsubstar{H}$ contains a graph on $\binom{k}{2}$ edges, which is necessarily $K_k$.

Linear combinations of the functions $\numindsubstar{H}$ for different $H$ may however incur cancellations that make $K_k$ or even all dense graphs disappear from the sub-expansion.
We use various techniques to prove lower bounds on the degrees of polynomials $q_{f,t}$ associated with general graph parameters $f(\star)$, and thus lower-bound the number of edges among the graphs in their sub-expansion. These techniques include several elementary facts from algebra, e.g., a multilinear version of the Schwartz-Zippel lemma~\cite{WilliamsWWY15,BjorklundDH15}, and the rich toolkit of \emph{Fourier analysis}~\cite{Jukna12,ODonnell14}.

\subsection{Our Results}

As outlined above, we can study the complexity of linear combinations $f(\star)$ of $k$-vertex pattern subgraph counts by showing that their polynomial representations
have large degree.
In the following, we discuss concrete results on the complexity of $\pindsub{\Phi}$ obtained this way.

\paragraph{Sparse and Hereditary Properties.}
We first show hardness of properties that hold only with probability $\exp(-\Omega(k^2))$ on $k$-vertex graphs.
This implies known results on hereditary properties~\cite{FockeR24} as a special case.
Moreover, our result holds even for ``weighted properties'':
In the following, a \emph{graph invariant} $\Phi$ is an isomorphism-invariant function from graphs to a field $\FF$. 
We call $\Phi$ a \emph{$k$-vertex graph invariant} if it is supported only on $k$-vertex graphs and define
\[\numindsub{\Phi}{G} = \sum_F \Phi(F) \cdot \numindsub{F}{G}.\]

The Schwartz-Zippel lemma~\cite{DemilloL78,Schwartz80,Zippel79} can be used to show that every non-zero function $f:\{0,1\}^m \to \FF$ of small support necessarily has large degree:
If $1 \leq |\supp(f)| \leq 2^{m-\ell}$ for $\ell \in \NN$, then the degree of the polynomial representation of $f$ over $\FF$ is at least $\ell$.
Using this, we show that the subgraph basis representation of $\numindsub{\Phi}{\star}$ contains a graph $H$ with $\ell$ edges, rendering this problem hard if $\ell$ is large.
In particular, when $\ell \geq \delta \cdot k^2$ for a constant $\delta > 0$, we obtain a graph $H$ with $\delta \cdot k^2$ edges, which implies tight lower bounds under ETH.
Since the Schwartz-Zippel lemma works over arbitrary fields, our results also apply to modular counting.

\begin{theorem}
    \label{thm:intro-hardness-small}
    For every $0 < \varepsilon < 1$ there is $\delta > 0$ such that, for every $k$-vertex graph invariant $\Phi$ with
    $1 \leq |\supp(\Phi)| \leq (2 - \varepsilon)^{\binom{k}{2}}$,
    no algorithm computes $\numindsub{\Phi}{G}$ in time $O(n^{\delta \cdot k})$ on $n$-vertex graphs $G$ unless ETH fails. 

    Moreover, for every computable graph invariant $\Phi$ such that
    $1 \leq |\supp(\slice{\Phi}{k})| \leq (2 - \varepsilon)^{\binom{k}{2}}$
    holds for infinitely many slices $\slice{\Phi}{k}$, the problem $\pindsub{\Phi}$ is $\sharpwone$-hard.
    This holds in particular when $\Phi$ is a hereditary property that is not meager.

    For every fixed prime $p$, the lower bounds hold even when the value $\numindsub{\Phi}{G}$ is only to be computed modulo $p$, replacing ETH by the randomized ETH and $\sharpwone$ by $\modwone{p}$.
\end{theorem}

The claimed consequences for hereditary properties follow from the fact that, for fixed graphs $F$ and large enough $k$, every non-empty property excluding induced $F$-subgraphs contains at most $(2-\varepsilon)^{\binom{k}{2}}$ graphs in the $k$-th slice, cf.~\cite{PromelS92} or a self-contained proof in Appendix~\ref{sec:F-avoid}.

\paragraph{Edge-Monotone Properties.}
Let $\Phi$ be a nontrivial $k$-vertex graph property that is monotone, i.e., closed under the deletion of edges.\footnote{In the literature, this condition is often called ``edge-monotone'' (see e.g.,  \cite{DoringMW24,FockeR24,RothSW24}).}
By adapting known results on the Fourier degrees of monotone functions with symmetry conditions~\cite{DodisK99,RivestV76}, 
we show that the subgraph expansion of $\Phi$ contains a graph $H$ that in turn contains a complete bipartite graph $K_{\ell,\ell}$ for $\ell > k/4$ as a subgraph.
This gives us particularly strong bounds under ETH.
As before, our results remain valid when the subgraph expansion is reduced modulo a prime $p$; then we obtain $\ell > k / p^2$.

\begin{theorem}
    \label{thm:intro-hardness-monotone}
    There is $\delta > 0$ such that, for every $k$-vertex graph property that is monotone and not $k$-trivial,
    no algorithm computes $\numindsub{\Phi}{G}$ in time $O(n^{\delta \cdot k})$ on $n$-vertex graphs $G$ unless ETH fails.
    Moreover, for every computable graph property $\Phi$ that, for infinitely many $k$, is monotone on $k$-vertex graphs and not $k$-trivial, the problem $\pindsub{\Phi}$ is $\sharpwone$-hard.

    For every fixed prime $p$, the lower bounds hold even when the value $\numindsub{\Phi}{G}$ is only to be computed modulo $p$, replacing ETH by the randomized ETH and $\sharpwone$ by $\modwone{p}$.
\end{theorem}

\paragraph{Fully Symmetric Properties.}
For \emph{fully symmetric} properties, i.e., properties that depend only on the number of edges, we obtain strong hardness by straightforward applications of known Fourier analysis results for fully symmetric functions.
Special cases of such properties were studied before~\cite{JerrumM15,JerrumM17,RothS20}, and it has been shown that non-trivial and fully symmetric graph properties $\Phi$ induce hard problems $\pindsub{\Phi}$~\cite{RothSW24}.
We obtain tight lower bounds under ETH with an arguably simpler proof.

\begin{theorem}
    \label{thm:intro-hardness-symmetric}
    There is $\delta > 0$ such that, for every $k$-vertex graph property $\Phi$ that is fully symmetric and not $k$-trivial, no algorithm computes $\numindsub{\Phi}{G}$ in time $O(n^{\delta \cdot k})$ on $n$-vertex graphs $G$ unless ETH fails.

    Moreover, for every computable graph property $\Phi$ that is fully symmetric on $k$-vertex graphs and not $k$-trivial for infinitely many $k$, the problem $\pindsub{\Phi}$ is $\sharpwone$-hard.
\end{theorem}

\paragraph{Properties Avoiding many Hamming Weights.}
We also consider properties $\Phi$ that avoid a superlinear number of Hamming weights, see~\cite{JerrumM15,RothSW24}.
Given a $k$-vertex graph invariant $\Phi$, we define its \emph{Hamming weight set} as
$\hw(\Phi) \coloneqq \{|E(H)| \mid H \in \Phi\}$,
i.e., the set of all edge-counts among the graphs in $\Phi$.
We then define a value $\beta_\Phi$ that equals $0$ when $\Phi$ is empty and $\binom{k}{2}+1-|\hw(\Phi)|$ otherwise.
For a general graph invariant $\Phi$ and $k\in\NN$, we define $\beta_\Phi(k) \coloneqq \beta_{\slice{\Phi}{k}}$.
By a symmetrization argument previously used in the study of Boolean functions~\cite{NisanS94}, we obtain:

\begin{theorem}
    \label{thm:intro-hardness-hw}
    There is $\delta > 0$ such that, for every $k$-vertex graph invariant $\Phi$,
    no algorithm computes $\numindsub{\Phi}{G}$ in time $O(n^{\delta \cdot \beta_\Phi/k})$ on $n$-vertex graphs $G$ unless ETH fails.
    Moreover, $\pindsub{\Phi}$ is $\sharpwone$-hard for computable graph invariants $\Phi$ with $\beta_\Phi(k) \in \omega(k)$.
\end{theorem}

\paragraph{Invariants with Small Image.}
Next, we consider a useful tool for $k$-vertex graph invariants $\Phi$ with a sufficiently small image.
We define $\iota_\Phi \coloneqq |\im(\Phi)| - 1$, where $\im(\Phi)$ denotes the set of all $a \in \FF$
attained as $\Phi(G)$ for some graph $G$.
Given an arbitrary mapping $\tau\colon \im(\Phi) \to \FF$, we can ``rename'' the images of $\Phi$ using $\tau$ to obtain a new invariant $\Phi \circ \tau$ defined via $(\Phi \circ \tau)(G) \coloneqq \tau(\Phi(G))$.

Suppose we can choose $\tau$ in such a way that the polynomial representation
of $\Phi \circ \tau$ has large degree $d$.
Since $\tau$ can be represented by a (univariate) polynomial of degree at most $\iota_\Phi$ using standard interpolation arguments, we can conclude that the polynomial representation
of $\Phi$ has degree at least $d/\iota_\Phi$.
In other words, if $\iota_\Phi$ is small and $\Phi \circ \tau$ has large degree $d$, then the large degree $d$ must stem from $\Phi$ itself, not from the composition by $\tau$.
In particular, if $\iota_\Phi$ is sufficiently small (compared to $d$), this implies that $\pindsub{\Phi}$ is $\sharpwone$-hard.

We apply this method to show $\sharpwone$-hardness for monotone graph invariants $\Phi$ whose images have sublinear size, i.e., $|\im(\slice{\Phi}{k})| = o(k)$.
This strengthens a result from \cite{DoringMW25}, which obtains hardness only if $|\im(\slice{\Phi}{k})| \leq (1 - \varepsilon) \cdot \sqrt{k}$.

\paragraph{The All-Even Property.}
Finally, we consider the curious property $\alleven$ containing exactly those graphs whose vertices all have even degree.
This property turns out to be interesting for several reasons. Firstly, it encodes a vector space over $\ZZ_2$, the cycle space of $K_k$, and we can show that precisely the bipartite graphs on $k$ vertices appear in its subgraph expansion.
This implies that $\pindsub{\alleven}$ is $\sharpwone$-hard, with tight lower bounds under ETH.

\begin{theorem}
    \label{thm:intro-hardness-alleven-small}
    There is $\delta > 0$ such that no algorithm computes $\numindsub{\alleven}{G}$ in time $O(n^{\delta \cdot k})$ on $n$-vertex graphs $G$ unless ETH fails.
    Moreover, $\pindsub{\alleven}$ is $\sharpwone$-hard.
\end{theorem}

Secondly, the fact that only bipartite graphs appear in the sub-expansion of $\alleven$ also yields non-trivial \emph{upper bounds} on its complexity.
This in turn also implies \emph{lower bounds} for other properties, as we can take pointwise products with $\alleven$ to filter out unwanted graphs from a property $\Phi$ and translate degree lower bounds and complexity implications for the filtered version back to $\Phi$; see Section~\ref{sec:even} for details.

\subsubsection*{Weisfeiler-Leman Dimension}

Our techniques allow us to analyze the Weisfeiler-Leman dimension of graph parameters.
The $k$-dimensional Weisfeiler-Leman algorithm ($k$-WL) is a standard heuristic for testing isomorphism of graphs with connections to various other areas~\cite{Grohe17,GroheN21,Kiefer20,MorrisLMRKGFB22}.
The \emph{Weisfeiler-Leman dimension} of a graph invariant $f$ is the minimal $k \geq 1$ such that $k$-WL distinguishes between all graphs $G,G'$ for which $f(G) \neq f(G')$.
This complexity measure captures inexpressibility in certain logics and lower bounds for symmetric computational models.

We bound the WL dimension of $\numindsubstar{\Phi}$ for various $k$-vertex graph invariants $\Phi$:
For example, for monotone $k$-vertex graph properties $\Phi$, the WL dimension of $\numindsubstar{\Phi}$ is at least $k/4$, and the WL dimension of $\numindsubstar{\slice{\alleven}{k}}$ is equal to $\left\lfloor k / 2 \right\rfloor$.

\subsection{Related Results}

Theorem~\ref{thm:intro-hardness-monotone} was independently proved by D{\"{o}}ring, Marx and Wellnitz \cite{DoringMW25}.
Their hardness proof also shows that, for monotone graph properties $\Phi$, we have $\altenum{\Phi}(H) \neq 0$ for suitable graphs $H$ containing a large bi-clique.
While we use polynomial representations, D{\"{o}}ring et al.\ use group-theoretic arguments to analyze the alternating enumerator, extending arguments from \cite{DoringMW24}.
As indicated above, \cite{DoringMW25} even obtains hardness for monotone graph invariants $\Phi$ assuming the image of $\Phi$ is sufficiently small.
Following the first publication of \cite{DoringMW25}, we realized that our methods can be used to obtain even stronger results for monotone invariants with small images.
We added those results to the paper (see Section \ref{sec:hardness-small}); we stress that the results were not present in the first version of this work which was published independently from \cite{DoringMW25}.

We also point out that \cite{CurticapeanDNW25}, which proves Theorem \ref{fact:density-makes-hard}, only appeared after the first publication of this work. 
Building on \cite{CurticapeanDNW25}, we obtain tight lower bounds under ETH for Theorems~\ref{thm:intro-hardness-small} and \ref{thm:intro-hardness-hw}.
Earlier versions of this work only established a lower bound of $\Omega(k / \sqrt{\log k})$ on the running time exponent.
In particular, the tight lower bounds in these theorems were obtained in \cite{CurticapeanDNW25} and should, at least in part, be attributed to that work.
In fact, these tight lower bounds were the main impetus for that paper.

\subsection{Organization of the Paper}

In Section~\ref{sec:prelim}, we give formal definitions of the graph-theoretic concepts used, brief introductions to parameterized and exponential-time complexity, and a synopsis of basic concepts from Boolean function analysis.
We describe a general framework for proving hardness of counting induced subgraphs based on polynomial representations in Section~\ref{sec:hardness-templates}.
Based on this foundation, we then prove our main hardness results in Section~\ref{sec:hardness-fourier} and the results on the all-even property in Section~\ref{sec:even}.
Finally, we adapt the proofs to obtain lower bounds for the Weisfeiler-Leman dimension of induced subgraph counts in Section~\ref{sec:wl}.

\section{Preliminaries}
\label{sec:prelim}

We write $\ZZ$ for the integers and $\ZZ_p = \{0,1,\dots,p-1\}$ for the ring with addition and multiplication taken modulo $p$.
For an integer $n \geq 1$, we define $[n] \coloneqq \{1,\dots,n\}$.
For a set $X$ and an integer $k \geq 0$, we use $\binom{X}{k}$ to denote the collection of $k$-element subsets of $X$.
For a tuple $\ba \in X^k$, we commonly write $a_i$ to denote the $i$-th entry of $\ba$, i.e., $\ba = (a_1,\dots,a_k)$.

\subsection{Graphs}

We use standard graph notation.
All graphs considered in this paper are undirected and simple, i.e., they contain no multiedges or self-loops.
We write $V(G)$ and $E(G)$ to denote the vertex set and edge set of a graph $G$, respectively.
Also, we use $vw$ as a shorthand for an edge $\{v,w\} \in E(G)$.

A graph $H$ is a \emph{subgraph} of another graph $G$, denoted $H \subseteq G$, if $V(H) \subseteq V(G)$ and $E(H) \subseteq E(G)$.
For a graph $G$ and a set $X \subseteq V(G)$, we write $G[X] = (X,\{vw \in E(G) \mid v,w \in X\})$ to denote the \emph{(vertex-)induced subgraph} of $G$ by $X$.
Similarly, for $S \subseteq E(G)$, we write $G[S] = (V(G),S)$ to denote the \emph{(edge-)induced subgraph} of $G$ by $S$.
It is usually clear from the context whether we are taking a vertex- or an edge-induced subgraph.

For $k \geq 1$, we write $K_k$ for the complete graph with $k$ vertices, i.e., $V(K_k) = [k]$ and $E(K_k) = \binom{[k]}{2}$.
For $k,\ell \geq 1$, we write $K_{k,\ell}$ to denote the complete bipartite graph with $k+\ell$ vertices, i.e., $V(K_{k,\ell}) = [k + \ell]$ and $E(K_{k,\ell}) = \{\{i,k+j\} \mid i \in [k], j \in [\ell]\}$.
We remark that we explicitly define the graph $K_k$ to have vertex set $[k]$.
In particular, for $S \subseteq \binom{[k]}{2}$, we commonly use $K_k[S]$ to refer to the graph with vertex set $[k]$ and edge set $S$.
If the index $k$ is clear from context, we sometimes omit it and simply write $K[S]$.

A \emph{graph invariant} $\Phi$ (over a field $\FF$) maps each graph $G$ to a value $\Phi(G) \in \FF$ such that $\Phi(G) = \Phi(G')$ for all isomorphic graphs $G,G'$.
In this work, we mostly consider graph invariants over the rationals $\QQ$ and over the $p$-element field $\ZZ_p$ for a prime $p$.
Unless stated otherwise, graph invariants are defined over the rationals.
A graph invariant is a \emph{graph property} if $\Phi(G) \in \{0,1\}$ for all graphs $G$.
For a graph $G$ and a graph property $\Phi$, we write $G \in \Phi$ as a shorthand for $\Phi(G) = 1$.
For $k \geq 1$, a \emph{$k$-vertex graph invariant/property} is a graph invariant/property $\Phi$ such that $\Phi(G) = 0$ if $|V(G)| \neq k$.
The \emph{$k$-th slice} $\slice{\Phi}{k}$ of a graph invariant $\Phi$ is the $k$-vertex graph invariant that agrees with $\Phi$ on $k$-vertex graphs (and is $0$ otherwise).
For $k \geq 1$, a graph invariant $\Phi$ (over $\FF$) is \emph{$k$-trivial} if there is some $c \in \FF$ such that $\Phi(G) = c$ for all $k$-vertex graphs $G$.
A graph invariant $\Phi$ is \emph{meager} if it is $k$-trivial for all but finitely many $k$.

For a graph invariant $\Phi$, we define
\[\numindsub{\Phi}{G} \coloneqq \sum_{X \subseteq V(G)} \Phi(G[X]).\]
Observe that, if $\Phi$ is a graph property, then $\numindsub{\Phi}{G}$ is the number of vertex-induced subgraphs of $G$ that satisfy $\Phi$.
We write $\numindsubstar{\Phi}$ for the map $G \mapsto \numindsub{\Phi}{G}$;
this is again a graph invariant (over the same field as $\Phi$).

\subsection{Parameterized Complexity}

We give a brief introduction to parameterized counting complexity and refer the reader to the textbooks~\cite{CyganFKLMPPS15,FlumG06} and theses~\cite{Curticapean15,Roth19} for additional background.

A \emph{parameterized counting problem} consists of a function $P\colon \Sigma^* \to \QQ$ and a \emph{parameterization} $\kappa \colon \Sigma^* \to \ZZ_{\geq 0}$.
For an instance $x \in \Sigma^*$, the number $\kappa(x)$ is called the \emph{parameter} of $x$.
The problem $(P,\kappa)$ is called \emph{fixed-parameter tractable (fpt)} if there is a computable function $f$, a constant $c \in \ZZ_{\geq 1}$ and a (deterministic) algorithm $\CA$ that, given an instance $x \in \Sigma^*$, computes $P(x)$ in time $f(\kappa(x)) \cdot |x|^{c}$.

A \emph{parameterized Turing reduction} from $(P,\kappa)$ to $(P',\kappa')$ is a deterministic fpt-algorithm $\CA$ with oracle access to $P'$ that, given an instance $x$, computes $P(x)$ and only makes oracle calls to instances $y$ that satisfy $\kappa'(y) \leq g(\kappa(x))$ for some computable function $g\colon \ZZ_{\geq 0} \to \ZZ_{\geq 0}$. 
The problem $\pclique$ takes as input a graph $G$ and an integer $k \geq 1$ and asks to determine the number of $k$-cliques in $G$.
The number $k$ is the parameter of the instance $(G,k)$.
A parameterized counting problem $(P,\kappa)$ is \emph{$\sharpwone$-hard} if there is a parameterized Turing reduction from $\pclique$ to $(P,\kappa)$.
A \emph{parsimonious} parameterized reduction is a parameterized Turing reduction that makes exactly one oracle call $y$ and then outputs the oracle response on $y$ without further processing.

For a graph invariant $\Phi$, the parameterized counting problem $\pindsub{\Phi}$ takes as input a graph $G$ and an integer $k \geq 1$ (the parameter) and outputs $\numindsub{\slice{\Phi}{k}}{G}$.
When studying the complexity of $\pindsub{\Phi}$, it will mostly suffice to restrict our attention to the \emph{$k$-th slice} of the problem, for fixed integers $k \geq 1$, where the input consists only of a graph $G$.
This motivates us to define the following problem:
Given $k \geq 1$ and a $k$-vertex graph invariant $\Phi$,
the problem $\sindsub{\Phi}$ takes as input a graph $G$ and outputs $\numindsub{\Phi}{G}$.
For a graph invariant $\Phi$, we also write $\kindsub{\Phi}{k}$ instead of $\sindsub{\slice{\Phi}{k}}$.
Compared to $\pindsub{\Phi}$, this problem is polynomial-time solvable for every fixed $\Phi$ and $k$.

\paragraph{Modular Counting.}

Given a prime $p$, we write $\#_p$ to denote the modular counting version of a problem.
For example, $\pmodclique{p}$ denotes the problem of counting $k$-cliques modulo $p$.
We say a problem $(P,\kappa)$ is $\modwone{p}$-hard if there is a parameterized Turing reduction from $\pmodclique{p}$ to $(P,\kappa)$.
Note that a parsimonious parameterized reduction from $\pclique$ to the non-modular counting version $(P',\kappa)$ of $P$ implies such a reduction.

If $\Phi$ is a graph invariant over $\ZZ_p$, then $\numindsub{\Phi}{G} \in \ZZ_p$ for every graph $G$, which leads to the computational problems $\pmodindsub{p}{\Phi}$ and $\smodindsub{p}{\Phi}$.
Observe that, if $\Phi$ is a property, we can meaningfully consider the modular counting versions $\pmodindsub{p}{\Phi}$ and $\smodindsub{p}{\Phi}$ for every prime~$p$, since $0,1\in \ZZ_p$.

\paragraph{Fine-Grained Lower Bounds.}

To obtain more fine-grained lower bounds on the complexity of $\pindsub{\Phi}$ as well as $\sindsub{\Phi}$, we also rely on the \emph{Exponential Time Hypothesis (ETH)} dating back to \cite{ImpagliazzoPZ01} (see, e.g., {\cite[Conjecture 14.1]{CyganFKLMPPS15}}).

\begin{conjecture}[Exponential Time Hypothesis (ETH)]
    There is a real number $\varepsilon > 0$ such that the problem $3$-\textsc{Sat} cannot be solved in time $O(2^{\varepsilon n})$, where $n$ denotes the number of variables of the input formula.
\end{conjecture}

In the setting of modular counting, we rely on the following slightly stronger conjecture that even rules out randomized algorithms.

\begin{conjecture}[Randomized Exponential Time Hypothesis (rETH)]
    There is a real number $\varepsilon > 0$ such that the problem $3$-\textsc{Sat} cannot be solved in time $O(2^{\varepsilon n})$ by a randomized algorithm with error probability at most $1/3$, where $n$ denotes the number of variables of the input formula.
\end{conjecture}

\subsection{Representing Boolean Functions}
\label{sec:boolean-functions}

We represent Boolean functions $f$ (that in turn represent graph parameters) as polynomials $p$ and lower-bound the degrees of these polynomials to obtain complexity-theoretic lower bounds.
In the following, we present the relevant theory for general Boolean functions; the specifics for graph parameters are discussed in Section~\ref{sec:sub-vs-fourier}.

\paragraph{Polynomial Representation.}
Let $\FF$ be a field.
Every Boolean function $f\colon \{0,1\}^m \to \FF$ of arity $m$ can be uniquely represented as a multilinear polynomial in $\FF[x_1,\ldots,x_m]$ via
\begin{equation}
    \label{eq: prelim-interpolation}
    q_f(\bx) = \sum_{\ba \in \{0,1\}^m} f(\ba) \prod_{i \in [m]:a_i = 1} x_i \prod_{i \in [m]:a_i = 0}(1 - x_i),    
\end{equation}
where $\bx = (x_1,\dots,x_m)$.
We call $q_f$ the \emph{polynomial representation of $f$}.
For $S \subseteq [m]$, let us denote the monomial corresponding to $S$ as
\[x^S \coloneqq \prod_{i \in S}x_i,\]
with the convention that $x^\emptyset = 1$.
By expanding \eqref{eq: prelim-interpolation} and collecting terms, we obtain unique coefficients $\altenum{f}(S) \in \FF$ for every $S \subseteq [m]$ such that 
\begin{equation}
    q_f(\bx) = \sum_{S \subseteq [m]} \altenum{f}(S) \cdot x^S.
\end{equation}
This notation is inspired by our notation for the alternating enumerator; see Definition~\ref{def:alt-enum} and Remark~\ref{rem:def-alternating-enumerator} for more on the origin of this notation.

We write $\deg_{\FF}(f)$ to denote the degree of the polynomial $q_f$, i.e.,
\[\deg_{\FF}(f) = \max\{|S| \mid S \subseteq [m] \text{ with } \altenum{f}(S) \neq 0\}.\]
If $\FF = \RR$, we simply omit the index and write $\deg(f)$, and we write $\deg_p(f)$ if $\FF = \ZZ_p$.
We are often dealing with functions $f\colon \{0,1\}^m \to \{0,1\}$ which can be interpreted over any field $\FF$.
In this situation, observe that $\deg_p(f) \leq \deg(f)$.

We regularly need to restrict Boolean functions to subcubes.
Let $f\colon \{0,1\}^m \to \FF$ be a Boolean function of arity $m$ and let $J \subseteq [m]$.
We write $\overline{J} \coloneqq [m] \setminus J$ for the complement of $J$ in $[m]$.
Also let $\bz = (z_j)_{j \in \overline{J}} \in \{0,1\}^{\overline{J}}$.
Then we define the \emph{restriction of $f$ to $J$ using $\bz$}, denoted by $f_{J|\bz}\colon\{0,1\}^{J} \to \FF$, to be the subfunction of $f$ given by fixing the coordinates in $\overline{J}$ to the values in $\bz$.
Slightly abusing notation, for $\by \in \{0,1\}^{J}$, we write $(\by,\bz)$ for the composite tuple in $\{0,1\}^{m}$, even tough $\by$ and $\bz$ are not actually concatenated.
In this notation, we have $f_{J|\bz}(\by) = f(\by,\bz)$.
Then
\begin{equation}
    \label{eq:fourier-restriction}
    \altenum{f_{J|\bz}}(S) = \sum_{T \subseteq \overline{J}} \altenum{f}(S \cup T) \prod_{j \in T}z_j
\end{equation}
for all $S \subseteq J$.
In particular, $\deg_\FF(f) \geq \deg_\FF(f_{J|\bz})$.

\paragraph{Fourier Representation.}
Boolean functions can also be represented by their \emph{Fourier expansion}.
In an earlier version of the paper, several proofs relied on this representation, while the current version instead works only in the representation over $\{0,1\}^m$ described above.
Nevertheless, we briefly discuss the Fourier representation here, because it can provide an alternative perspective on some of our results, and it originally led us to their discovery. Pragmatic readers and readers familiar with Fourier analysis can safely skip this part.

For simplicity, we restrict our attention to functions that map into $\RR$.
In Fourier analysis, we consider the input space of Boolean functions to be $\{-1,1\}^m$, and the reader is encouraged to associate $-1$ with \True\ and $1$ with \False.
Let $f\colon \{-1,1\}^m \to \RR$.
As before, there is a unique multilinear polynomial in $\RR[x_1,\ldots,x_m]$ representing $f$, i.e., there are unique so-called \emph{Fourier coefficients} $\widehat{f}(S) \in \RR$ for $S \subseteq [m]$ with 
\begin{equation}
    f(\bx) = \sum_{S \subseteq [m]} \widehat{f}(S) \cdot x^S.
\end{equation}
The polynomial representation of $f$ is called the \emph{Fourier expansion of $f$}.
We refer to the textbook \cite{ODonnell14} for extensive background on the Fourier expansion of Boolean functions.
We write
\[
\begin{aligned}
\supp(f) & \coloneqq \{\bx \in \{-1,1\}^m \mid f(\bx) \neq 0\},\\
\supp(\widehat{f}) & \coloneqq \{S \subseteq [m] \mid \widehat{f}(S) \neq 0\}.
\end{aligned}
\]
By the \emph{uncertainty principle}, one of these two sets is always large; see \cite[Exercise 3.15]{ODonnell14}.

\begin{lemma}[Uncertainty Principle]
    \label{lem:uncertainty}
    If $f\colon \{-1,1\}^m \to \RR$ is not the zero function, then
    \[|\supp(f)| \cdot |\supp(\widehat{f})| \geq 2^m.\]
\end{lemma}

The \emph{Fourier degree} $\deg(f)$ of $f$ is the degree of the Fourier expansion, i.e.,
\[\deg(f) \coloneqq \max\{|S| \mid S \in \supp(\widehat f) \}.\]

It is possible to change between the $\{0,1\}$-based and $\{-1,1\}$-based representations by a simple affine shift of variables. More concretely, to every function $f\colon \{0,1\}^m \to \RR$, we can canonically associate a function $f^{*}\colon \{-1,1\}^m \to \RR$
by setting $f^*(\ba^*) \coloneqq f(\ba)$, where $\ba$ is obtained from $\ba^*$ by replacing every $1$ by $0$, and every $-1$ by $1$.
We have that
\begin{equation}
    \label{eq:poly-basis-change}
    f^*(\bx)=\sum_{S \subseteq [m]} \widehat{f^*}(S) x^S = q_f\left(-\frac{1}{2}x_1+\frac{1}{2},\dots,-\frac{1}{2}x_m+\frac{1}{2}\right).
\end{equation}

As the right-hand side equals $q_f$ up to an affine shift of variables, we obtain $\deg(f) = \deg(f^*)$.
Even stronger, $q_f$ and the polynomial representing $f^*$ have the same inclusion-wise maximal monomials.

\section{General Approach for Hardness Proofs}
\label{sec:sub-vs-fourier}

While previous works (see, e.g., \cite{DorflerRSW22,DoringMW24,RothS20}) analyze the alternating enumerator directly to obtain hardness results, we observe a close connection to the coefficients of the polynomial representation of $f_\Phi$, i.e., the Boolean function associated with $\Phi$.
This allows us to use algebraic techniques to obtain hardness results for $\sindsub{\Phi}$ for various different types of graph invariants $\Phi$.

\subsection{Polynomial Representations and Alternating Enumerators}

Given a $k$-vertex graph invariant $\Phi$, we can represent $\Phi$ canonically as a Boolean function.
For ease of notation, we write $K \coloneqq K_k$ for the complete graph on vertex set $[k]$ in the remainder of this section.
We associate to $\Phi$ the Boolean function $f_{\Phi} \colon \{0,1\}^{E(K)} \to \FF$ of arity $\binom{k}{2}$, where  
\[f_{\Phi}(\ba) \coloneqq \Phi(K[\ba])\]
with $K[\ba] \coloneqq K[S_{\ba}]$ and $S_{\ba} = \{e \in E(K) \mid a_e = 1\}$.

In the following, we write $\bx = (x_e)_{e \in E(K)}$.
As outlined in Section~\ref{sec:boolean-functions}, there is a unique multilinear polynomial in $\FF[\bx]$ that agrees with $f_\Phi$, the \emph{polynomial representation of $\Phi$}, given by
\begin{equation}
    \label{eq:def-q-phi}
    q_\Phi(\bx) = \sum_{\ba \in \{0,1\}^{E(K)}} f_\Phi(\ba)\prod_{e \in E(K) \colon a_e = 1} x_e \prod_{e\in E(K) \colon a_e = 0}(1 - x_e).
\end{equation}
Since $\Phi$ is a graph invariant, it attains the same value on graphs with the same isomorphism type. This means we can collect terms in \eqref{eq:def-q-phi} and change notation, to obtain
\begin{equation}
\label{eq:qphi-indsub-poly}    
q_\Phi(\bx) = \sum_{\substack{\text{unlabeled }H \\\text{on $k$ vertices}}} \Phi(H)\underbrace{\sum_{\substack{F \subseteq K \\F \cong H}}  \prod_{e\in E(F)} x_e \prod_{e\notin E(F)}(1 - x_e)}_{=\indsubpoly{H,k}(\bx)}\,,
\end{equation}
with $\indsubpoly{H,k}$ from \eqref{eq:indsub-polynomial}.
As we only require the polynomials $\indsubpoly{H,k}$ rather than $\indsubpoly{H,t}$ for general $t \in \NN$, we abbreviate $\indsubpoly{H} \coloneqq \indsubpoly{H,k}$.
Similarly, we abbreviate $\subpoly{H} \coloneqq \subpoly{H,k}$, where the latter was defined in \eqref{eq:sub-polynomial}.

We wish to understand the polynomial representation of $\Phi$ in more detail and connect it to graph-theoretic notions. First, recall the following combinatorial definition:
\begin{definition}
\label{def:alt-enum}
    The \emph{alternating enumerator} of a graph invariant $\Phi$ on a graph $H$ is
    \[\altenum{\Phi}(H) = (-1)^{|E(H)|} \sum_{S \subseteq E(H)} (-1)^{|S|} \,\Phi(H[S]).\]
\end{definition}
\begin{remark}
    \label{rem:def-alternating-enumerator}
    In earlier works~\cite{DorflerRSW22,DoringMW24,RothS20}, the alternating enumerator is defined without the factor $(-1)^{|E(H)|}$.
    These works aim at showing that $\altenum{\Phi}(H) \neq 0$ for certain graphs $H$; such statements are not affected by the absence of the factor $(-1)^{|E(H)|}$.
    To obtain more straightforward statements in Lemma~\ref{lem:polynomial-alternating-enumerator} and Corollaries \ref{cor:polynomial-via-sub-poly} and \ref{cor:phi=sub}, we prefer the variant defined above.

    Additionally, earlier works usually write $\widehat{\Phi}(H)$ to denote the alternating enumerator.
    We avoid this notation, since the hat-notation $\widehat{f}$ is commonly used in the analysis of Boolean functions to denote Fourier coefficients (see, e.g., \cite{ODonnell14}).
    Instead, our notation follows \cite{Lovasz12}.
\end{remark}

By expanding the right-hand side of \eqref{eq:def-q-phi} similar to $\eqref{eq:sub-polynomial}$, we see that the polynomial representation of $\Phi$ has the alternating enumerators as coefficients in its monomial expansion.

\begin{lemma}
    \label{lem:polynomial-alternating-enumerator}
    Let $\Phi$ be a graph invariant on $k$-vertex graphs.
    Then
    \begin{equation}
        \label{eq:alternating-enumerator-polynomial}
        q_\Phi(\bx) = \sum_{S \subseteq E(K)} \altenum{\Phi}(K[S]) \cdot x^S,
    \end{equation}
    i.e., $\altenum{f_\Phi}(S) = \altenum{\Phi}(K[S])$ for every $S \subseteq E(K)$.
\end{lemma}

\begin{proof}
    For $E \subseteq E(K)$ let $q_E(\bx) \coloneqq \prod_{e \in E} x_e \prod_{e \in E(K) \setminus E} (1 - x_e)$.
    Then
    \[q_E(\bx) = \sum_{S \subseteq E(K) \setminus E} (-1)^{|S|} x^{E \cup S} = \sum_{S \supseteq E} (-1)^{|S \setminus E|} x^{S}.\]
    It follows that the coefficient of $x^S$ is equal to
    \[\sum_{E \subseteq S} (-1)^{|S \setminus E|} \cdot \Phi(K[E]) = (-1)^{|S|} \cdot \sum_{E \subseteq S} (-1)^{E} \cdot \Phi(K[E]) = \altenum{\Phi}(K[S]).\]
\end{proof}

As a direct corollary, a large degree $\deg(q_\Phi) = \deg(f_\Phi)$ of $f_\Phi$ implies the existence of graphs $H$ with non-zero alternating enumerator and many edges.

\begin{corollary}
    \label{cor:alt-enum-vs-degree}
    Let $k \geq 1$ and let $\Phi$ be a $k$-vertex graph invariant.
    Then $\deg(q_\Phi) \geq \ell$ if and only if there is a $k$-vertex graph $H$ such that $|E(H)| \geq \ell$ and $\altenum{\Phi}(H) \neq 0$.
\end{corollary}

By collecting terms corresponding to isomorphic graphs, the polynomial $q_\Phi$ can be expressed as a linear combination of the subgraph polynomials defined in \eqref{eq:sub-polynomial}.

\begin{corollary}
    \label{cor:polynomial-via-sub-poly}
    Let $\Phi$ be a graph invariant on $k$-vertex graphs.
    Then
    \begin{equation}
        \label{eq:alternating-enumerator-polynomial-compact}
        q_\Phi(\bx) = \sum_{\substack{\text{unlabeled }H \\\text{on $k$ vertices}}} \altenum{\Phi}(H) \cdot \subpoly{H}(\bx).
    \end{equation}
\end{corollary}

\begin{proof}
    If $\Phi$ is a graph invariant, then also $\altenum{\Phi}$ is a graph invariant by definition.
    Therefore, we have $\altenum{\Phi}(K[S]) = \altenum{\Phi}(K[S'])$ if $K[S]\cong K[S']$.
    This allows us to collect terms corresponding to isomorphic graphs in \eqref{eq:alternating-enumerator-polynomial} to obtain
    \[
    q_\Phi(\bx) 
    = \sum_{S \subseteq E(K)} \altenum{\Phi}(K[S]) \cdot x^S
    = \sum_{\substack{\text{unlabeled }H \\\text{on $k$ vertices}}} \altenum{\Phi}(H)
    \underbrace{\sum_{\substack{S \subseteq E(K) \\ K[S] \cong H}} x^S}_{=\subpoly{H}(\bx)},\]
    which proves the corollary.
\end{proof}

In particular, \eqref{eq:qphi-indsub-poly} and \eqref{eq:alternating-enumerator-polynomial} together imply that $\Phi(F) = \sum_{H} \altenum{\Phi}(H) \cdot \numsub{H}{F}$ for every $k$-vertex graph $F$.
We show below that this holds even when $|V(F)| \neq k$.
Related correspondences have been shown combinatorially, see~\cite[Lemma~8]{DorflerRSW22} or \cite[Theorem~1]{RothS20}.

\begin{corollary}
    \label{cor:phi=sub}
    Given a $k$-vertex graph invariant $\Phi$, for $k \geq 1$, we have
    \begin{equation}
        \numindsub{\Phi}{\star} = \sum_{\substack{\text{unlabeled }H \\\text{on $k$ vertices}}} \altenum{\Phi}(H) \cdot \numsub{H}{\star}.
    \end{equation}
\end{corollary}
\begin{proof}
    Given an $n$-vertex graph $G$, we can write $\numindsub{\Phi}{G} = \sum_X \numindsub{\Phi}{G[X]}$, where $X$ ranges over all $k$-vertex subsets of $V(G)$.
    Since $G[X]$ is a $k$-vertex graph, we have $\numindsub{\Phi}{G[X]} = \Phi(G[X]) = \sum_H \altenum{\Phi}(H)\cdot \numsub{H}{G[X]}$.
    Changing the order of summation, we obtain 
    \[\numindsub{\Phi}{G} = \sum_H \altenum{\Phi}(H)\sum_X \numsub{H}{G[X]} = \sum_H \altenum{\Phi}(H) \cdot \numsub{H}{G}.\]
\end{proof}

\subsection{Non-Vanishing Alternating Enumerators Imply Hardness}
\label{sec:hardness-templates}

To show hardness of the problem $\sindsub{\Phi}$ and its variants, we follow previous works \cite{DorflerRSW22,DoringMW24,RothS20} and identify graphs of unbounded treewidth with non-zero alternating enumerator:
For a graph $H$, we write $\tw(H)$ for the treewidth of $H$ (see, e.g., \cite[Chapter 7]{CyganFKLMPPS15}) and say that an (infinite) sequence of graphs $(H_k)_{k \geq 1}$ has \emph{unbounded treewidth} if, for every $t \geq 1$, there is some $k \geq 1$ such that $\tw(H_k) \geq t$.

\begin{theorem}
    \label{thm:hardness-w}
    Let $(H_k)_{k \geq 1}$ be a sequence of graphs of unbounded treewidth.
    \begin{enumerate}[label = (\alph*)]
        \item\label{item:hardness-w-1} If $\Phi$ is a computable graph invariant with $\altenum{\Phi}(H_k) \neq 0$ for all $k \geq 0$, then $\pindsub{\Phi}$ is $\sharpwone$-hard.
        \item\label{item:hardness-w-2} If $p$ is prime and $\Phi$ is a computable graph invariant over $\ZZ_p$ such that $\altenum{\Phi}(H_k) \neq 0 \bmod p$ for all $k \geq 0$, then $\pmodindsub{p}{\Phi}$ is $\modwone{p}$-hard.
    \end{enumerate}
\end{theorem}

The next theorem gives us more fine-grained lower bounds based on ETH and rETH.

\begin{theorem}
    \label{thm:hardness-eth}
    There is a universal constant $\alpha_{\textsc{ind}} > 0$ and an integer $N_0 \geq 1$ such that for all numbers $k,\ell \geq 1$, the following holds:
    \begin{enumerate}[label = (\alph*)]
        \item\label{item:hardness-eth-1} If $\Phi$ is a $k$-vertex graph invariant and there exists a graph $H$ with $\altenum{\Phi}(H) \neq 0$ and $E(H) \geq k \cdot \ell \geq N_0$, then $\sindsub{\Phi}$ cannot be solved in time $O(n^{\alpha_{\textsc{ind}} \cdot \ell})$ unless ETH fails.
        \item\label{item:hardness-eth-3} If $p$ is prime, $\Phi$ is a $k$-vertex graph invariant over $\ZZ_p$ and there exists a graph $H$ with $\altenum{\Phi}(H) \neq 0 \bmod p$ and $E(H) \geq k \cdot \ell \geq N_0$, then $\smodindsub{p}{\Phi}$ cannot be solved in time $O(n^{\alpha_{\textsc{ind}} \cdot \ell})$ unless rETH fails.
    \end{enumerate}
\end{theorem}

Both theorems can be proved by (variants of) known arguments \cite{Curticapean15,DorflerRSW22,DoringMW24,RothS20} via Corollary~\ref{cor:phi=sub}.
Still, we provide full proofs in Appendix \ref{app:omitted-proofs}.

Finally, let us also note the following algorithmic result.
Let $G$ be a graph.
A \emph{vertex cover} of $G$ is a set $C \subseteq V(G)$ such that $e \cap C \neq \emptyset$ for every edge $e \in E(G)$.
The \emph{vertex cover number} of $G$, denoted by $\vc(G)$, is the minimal size of a vertex cover of $G$.

\begin{theorem}
    \label{thm:algorithm-alternating-enumerator-vc}
    Let $\Phi$ be a computable graph invariant and let $t(k)$ denote the maximal vertex cover number of a $k$-vertex graph $H$ such that $\altenum{\Phi}(H) \neq 0$.
    Then $\pindsub{\Phi}$ can be solved in time $f(k) \cdot n^{t(k) + 1}$ for some computable function $f$.
\end{theorem}

The theorem is a consequence of \cite[Theorem~1.1]{CurticapeanDM17} and Corollary~\ref{cor:phi=sub}; a proof is again given in Appendix \ref{app:omitted-proofs}.

\section{Hardness Results}
\label{sec:hardness-fourier}

Section \ref{sec:hardness-templates} shows that graphs $H$ with many edges and non-zero alternating enumerators $\altenum{\Phi}(H)$ imply hardness of $\sindsub{\Phi}$.
By Corollary~\ref{cor:alt-enum-vs-degree}, such graphs $H$ exist if the polynomial representation of $\Phi$ has large degree.
This allows us to exploit tools from the analysis of Boolean functions to find such graphs $H$.
Thus, combining the generic hardness results with algebraic techniques, we obtain hardness results for specific problems $\sindsub{\Phi}$ and $\pindsub{\Phi}$.

\subsection{Moderately Sparse Invariants}

As a first application, we consider graph invariants with moderately small support.
For a $k$-vertex graph invariant $\Phi$ over a field $\FF$, we write
\[\supp(\Phi) \coloneqq \Big\{K_k[S] ~\Big|~ S \in \binom{[k]}{2}, \Phi(K_k[S]) \neq 0\Big\}.\]
Observe that $\supp(\Phi)$ contains exactly the graphs $G$ with vertex set $[k]$ for which $\Phi(G) \neq 0$.
For an arbitrary graph invariant $\Phi$ (which is not only defined on $k$-vertex graphs) we also write $\supp_k(\Phi) \coloneqq \supp(\slice{\Phi}{k})$.

\begin{theorem}
    \label{thm:alternating-enumerator-small}
    Let $\FF$ be a field.
    Let $k \geq 1$, $0 \leq \ell \leq \binom{k}{2}$ and let $\Phi$ be a $k$-vertex graph invariant over $\FF$ such that
    \[1 \leq |\supp(\Phi)| \leq 2^{\binom{k}{2} - \ell}.\]
    Then there is a $k$-vertex graph $H$ such that $\altenum{\Phi}(H) \neq 0$ and $|E(H)| \geq \ell$.
\end{theorem}

We stress that Theorem~\ref{thm:alternating-enumerator-small} holds over arbitrary fields $\FF$.
The proof relies on the following lemma on Boolean functions, which can be viewed as a variant of the Schwartz-Zippel lemma for multlinear polynomials, and which was shown, e.g., in \cite[Lemma~2.2]{WilliamsWWY15} and \cite[Lemma~1]{BjorklundDH15}.

\begin{lemma}
    \label{lem:small-support-large-degree}
    Let $f\colon \{0,1\}^m \to \FF$ be a Boolean function that is not the zero function.
    Then
    \[|\{\ba \in \{0,1\}^m \mid f(\ba) \neq 0\}| \geq 2^{m - d}\]
    where $d \coloneqq \deg_\FF(f)$.
\end{lemma}

\begin{proof}
    Let $q_f \in \FF[\bx]$ be the polynomial representation of $f$ in indeterminates $\bx = (x_1,\dots,x_m)$ and let $S \subseteq [m]$ be a set of size $d$ such that $\altenum{f}(S) \neq 0$, i.e., the monomial $x^S$ has a non-zero coefficient in $q_f$.
    Without loss of generality assume that $S = \{1,\dots,d\}$.
    For every $\bc \in \{0,1\}^{\overline{S}}$, consider the polynomial
    \[q_{\bc}(x_1,\dots,x_d) = q_f(x_1,\dots,x_d,c_{d+1},\dots,c_{m})\]
    obtained from $q_f$ by substituting $\bc = (c_{d+1},\dots,c_{m})$ into $x_{d+1},\ldots,x_m$.
    Note that $q_{\bc}$ is not the zero polynomial, since the coefficient of $x^{S}$ remains unchanged.
    Because the multilinear polynomial $q_{\bc}$ is determined by its evaluations on all tuples in $\{0,1\}^S$, there is some $\bb \in \{0,1\}^S$ such that $q_{\bc}(\bb) \neq 0$.

    So overall, for every $\bc \in \{0,1\}^{\overline{S}}$ there is some $\bb \in \{0,1\}^S$ such that $f(\bb,\bc) = q_{\bc}(\bb) \neq 0$.
    It follows that there are at least $|\{0,1\}^{\overline{S}}| = 2^{m-d}$ vectors $\ba \in \{0,1\}^m$ with $f(\ba) \neq 0$.
\end{proof}

\begin{proof}[Proof of Theorem~\ref{thm:alternating-enumerator-small}]
    Let $k \geq 1$, $0 \leq \ell \leq \binom{k}{2}$ and let $\Phi$ be a $k$-vertex graph invariant over $\FF$ such that
    \[1 \leq |\supp(\Phi)| \leq 2^{m-\ell}\]
    where $m \coloneqq \binom{k}{2}$.
    Since $|\supp(\Phi)| \geq 1$ we get that $f_{\Phi}\colon \{0,1\}^{E(K_k)} \to \FF$ is not the zero function.
    Let $d \coloneqq \deg_\FF(f_\Phi)$.
    Then, by Lemma~\ref{lem:small-support-large-degree}, we get that
    \[2^{m - d} \leq |\{\ba \in \{0,1\}^{E(K_k)} \mid f_\Phi(\ba) \neq 0\}| = |\supp(\Phi)| \leq 2^{m - \ell}.\]
    It follows that $d \geq \ell$.
    By Corollary~\ref{cor:alt-enum-vs-degree}, we obtain a $k$-vertex graph $H$ such that $\altenum{\Phi}(H) \neq 0$ and $|E(H)| \geq \ell$.
\end{proof}

\begin{remark}
    An earlier version of the paper showed a slightly weaker form of Theorem \ref{thm:alternating-enumerator-small} for $\FF = \RR$ via the uncertainty principle (see Lemma \ref{lem:uncertainty}).
    Indeed, if $\supp(\Phi)$ is small, then the Fourier representation of $\Phi$ has large support, which implies its degree is also large.
\end{remark}

By combining Theorem \ref{thm:alternating-enumerator-small} for $\FF = \QQ$ with Theorems \ref{thm:hardness-w}\ref{item:hardness-w-1} and \ref{thm:hardness-eth}\ref{item:hardness-eth-1}, we obtain the following.

\begin{corollary}
    \label{cor:hardness-small}
    For every $0 < \varepsilon < 1$ there are $N_0,\delta > 0$ such that the following holds.
    Let $k \geq N_0$ and let $\Phi$ be a $k$-vertex graph invariant such that
    \[1 \leq |\supp(\Phi)| \leq (2 - \varepsilon)^{\binom{k}{2}}.\]
    Then no algorithm solves $\sindsub{\Phi}$ in time $O(n^{\delta \cdot k})$ unless ETH fails.
    
    Moreover, for every computable graph invariant $\Phi$ such that
    \[1 \leq |\supp_k(\Phi)| \leq (2 - \varepsilon)^{\binom{k}{2}}\]
    holds for infinitely many $k$, the problem $\pindsub{\Phi}$ is $\sharpwone$-hard.
\end{corollary}

Also, combining Theorem \ref{thm:alternating-enumerator-small} for $\FF = \ZZ_p$ with Theorems \ref{thm:hardness-w}\ref{item:hardness-w-2} and \ref{thm:hardness-eth}\ref{item:hardness-eth-3}, we obtain hardness for modular counting.

\begin{corollary}
    \label{cor:hardness-small-mod}
    For every $0 < \varepsilon < 1$ there are $N_0,\delta > 0$ such that the following holds.
    Let $p$ be a prime, $k \geq N_0$ and let $\Phi$ be a $k$-vertex graph invariant over $\ZZ_p$ such that
    \[1 \leq |\supp(\Phi)| \leq (2 - \varepsilon)^{\binom{k}{2}}.\]
    Then no algorithm solves $\smodindsub{p}{\Phi}$ in time $O(n^{\delta \cdot k})$ unless rETH fails.
    
    Moreover, for every computable graph invariant $\Phi$ over $\ZZ_p$ such that
    \[1 \leq |\supp_k(\Phi)| \leq (2 - \varepsilon)^{\binom{k}{2}}\]
    holds for infinitely many $k$, the problem $\pmodindsub{p}{\Phi}$ is $\modwone{p}$-hard.
\end{corollary}

Together, Corollary \ref{cor:hardness-small} and \ref{cor:hardness-small-mod} restate Theorem \ref{thm:intro-hardness-small}.

A notable example of graph properties of moderately small support are hereditary graph properties.
A graph property $\Phi$ is \emph{hereditary} if it is closed under taking (vertex-)induced subgraphs: If $G \in \Phi$, then $G[X] \in \Phi$ for every $X \subseteq V(G)$.

\begin{theorem}[\cite{PromelS92}]
    Let $\Phi$ be a non-trivial hereditary graph property.
    Then there is an integer $k_0 \geq 1$ and $\varepsilon > 0$ such that
    \[|\supp_k(\Phi)| \leq (2 - \varepsilon)^{\binom{k}{2}}\]
    for every $k \geq k_0$.
\end{theorem}

\begin{proof}
    Since $\Phi$ is non-trivial, there is some graph $F$ with $\Phi(F) = 0$.
    Hence, every graph $G \in \Phi$ does not contain $F$ as an induced subgraph.
    Now the statement follows directly from \cite{PromelS92} (see also Theorem~\ref{thm:Fcopies} in the appendix\footnote{To keep the paper self-contained, we include a simple proof of a weaker version of the main result in \cite{PromelS92} as Theorem~\ref{thm:Fcopies}. This weaker version is sufficient for our purposes and can be obtained with significantly less technical effort.}) that bounds the number of graphs that do not contain $F$ as an induced subgraph.
\end{proof}

In particular, Corollary \ref{cor:hardness-small} generalizes the results from \cite{FockeR24} on $\sharpwone$-hardness and tight lower bounds under ETH for non-trivial hereditary properties. Note that Corollary \ref{cor:hardness-small} also covers many non-hereditary properties and even holds for graph invariants that are not necessarily $0$-$1$-valued. In particular, negative weights are allowed, which could \emph{a priori} lead to intricate cancellations.
On top of that, Corollary~\ref{cor:hardness-small} yields hardness for modular counting.

\subsection{Monotone Properties}
\label{sec:hardness-monotone}

Next, we concern ourselves with monotone properties.
Let $k \geq 1$ and let $\Phi$ be a $k$-vertex graph property.
We say that $\Phi$ is \emph{monotone}\footnote{In the literature on counting induced subgraphs (see, e.g., \cite{DoringMW24,FockeR24,RothSW24}), this condition is usually called \emph{edge-monotone}, whereas \emph{monotone} properties are those that are hereditary and edge-monotone (on $k$-vertex graphs for all $k \geq 1$).
In this paper, we use the term \emph{monotone} rather than \emph{edge-monotone}, to ensure consistency with notation used for analysing Boolean functions (see, e.g., \cite{ODonnell14}).} if $G[S] \in \Phi$ for every $G \in \Phi$ and every $S \subseteq E(G)$.
For an arbitrary property $\Phi$, we say that $\Phi$ is \emph{monotone on $k$-vertex graphs} if $\slice{\Phi}{k}$ is monotone. 

The following theorem is a simple modification of a result from \cite{DodisK99}.

\begin{theorem}
    \label{thm:alternating-enumerator-monotone}
    Let $k \geq 1$ and $p$ a prime number.
    Let $\Phi$ be a $k$-vertex graph property that is monotone and not $k$-trivial.

    Then there is some $k$-vertex graph $H$ such that $\altenum{\Phi}(H) \neq 0 \bmod p$ and the complete bipartite graph $K_{\ell,\ell}$ is a subgraph of $H$ for some $\ell > \frac{k}{p^2}$.
\end{theorem}

More concretely, building on earlier work of Rivest and Vuillemin \cite{RivestV76}, Dodis and Khanna \cite{DodisK99} show that the polynomial representation $q_{\Phi}$ of $\Phi$ (see \eqref{eq:def-q-phi}) has degree $\Omega(k^2)$ if $\Phi$ is a graph property on $k$-vertex graphs that is monotone and not $k$-trivial.
Actually, \cite{DodisK99} shows that $q_{\Phi}$ even has degree $\Omega(k^2)$ when interpreted as a polynomial over the two-element field $\ZZ_2$, i.e., all coefficients of the polynomial $q_{\Phi}$ are taken modulo $2$.
Analyzing their arguments in more detail, it is in fact shown that $q_{\Phi}$ contains a monomial (with non-zero coefficient) that corresponds to a graph $H$ such that $K_{\ell,\ell}$ is a subgraph of $H$ for some $\ell > \frac{k}{4}$.
A simple extension of the arguments gives the statement above for every prime $p \geq 3$.

Since \cite{DodisK99} does not provide proof details and their arguments need to be slightly extended, we give a full proof here.
Let us stress again that the following proof of Theorem \ref{thm:alternating-enumerator-monotone} precisely follows the arguments from \cite{DodisK99}.

Toward the proof of Theorem \ref{thm:alternating-enumerator-monotone}, we require some additional tools.
For a finite set $\Omega$ we write $\Sym(\Omega)$ to denote the symmetric group over ground set $\Omega$.
For $m \geq 1$ we also write $S_m$ as a shortcut for $\Sym([m])$.

Let $p$ be a prime number and $m \geq 1$.
For the remainder of this subsection, we consider Boolean functions $f\colon \{0,1\}^m \to \ZZ_p$ that map into the $p$-element field $\ZZ_p$.
A permutation $\sigma \in S_m$ is an \emph{automorphism of $f$} if $f(\ba) = f(\ba^\sigma)$ for every $\ba = (a_1,\dots,a_m) \in \{0,1\}^m$ where $\ba^{\sigma} = (a_{\sigma(1)},\dots,a_{\sigma(m)})$ is the tuple obtained from $\ba$ by permuting its entries according to $\sigma$.
We define the \emph{automorphism group of $f$}, denoted by $\Aut(f)$, to be the group of all automorphisms of $f$.
We say that $f$ is \emph{transitive} if $\Aut(f)$ is transitive, i.e., for every $i,j \in [m]$ there is some $\sigma \in \Aut(f)$ such that $\sigma(i) = j$.

Recall that a Boolean functions $f\colon \{0,1\}^m \to \ZZ_p$ can be written as a multilinear polynomial over $\ZZ_p$ as
\[q_f(\bx) = \sum_{S \subseteq [m]} \altenum{f}(S) \cdot x^S\]
for uniquely determined coefficients $\altenum{f}(S) \in \ZZ_p$ for $S \subseteq [m]$.
We write $\deg_p(f)$ to denote the degree of $q_f$ (over the field $\ZZ_p$).

Also, we write $\vec 0 = (0,\dots,0)$ and $\vec 1 = (1,\dots,1)$ to denote the all-zero and all-one vector of the appropriate length, respectively (the length of the vector is always clear from context).

\begin{lemma}
    \label{lem:full-degree-prime-power}
    Let $f\colon \{0,1\}^m \to \ZZ_p$ be a function over $\ZZ_p$ where $p$ is prime and $m = p^k$ for some integer $k \geq 1$.
    Also suppose $f$ is transitive and $f(\vec 0) \neq f(\vec 1)$.
    Then $\deg_p(f) = m$.
\end{lemma}

\begin{proof}
    First observe that $f(\vec 0) = q_f(\vec 0) = \altenum{f}(\emptyset)$ and
    \[f(\vec 1) = q_f(\vec 1) = \sum_{S \subseteq [m]} \altenum{f}(S).\]
    We need to argue that $\altenum{f}([m]) \neq 0$.
    The following two claims form the key ingredients.
    \begin{claim}
        \label{claim:coefficients-equal-in-orbit}
        For every $S \subseteq [m]$ and $\sigma \in \Aut(f)$, we have $\altenum{f}(S) = \altenum{f}(S^\sigma)$, where $S^{\sigma} = \{\sigma(i) \mid i \in S\}$.
    \end{claim}
    \begin{claimproof}
        Let $\sigma \in \Aut(f)$.
        Consider the polynomial
        \[q_f^{\sigma}(\bx) = \sum_{S \subseteq [m]} \altenum{f}(S^{\sigma}) \cdot x^S.\]
        We have that
        \[q_f(\ba) = f(\ba) = f(\ba^{\sigma}) = q_f^{\sigma}(\ba)\]
        for all $\ba \in \{0,1\}^m$.
        Since the polynomial representation of $f$ is unique, we conclude that $q_f = q_f^{\sigma}$ which implies that $\altenum{f}(S) = \altenum{f}(S^\sigma)$ for every $S \subseteq [m]$.
    \end{claimproof}

    \begin{claim}
        \label{claim:p-divides-orbit-size}
        For every $\emptyset \subsetneq S \subsetneq [m]$, the cardinality of $S^{\Aut(f)} \coloneqq \{S^\sigma \mid \sigma \in \Aut(f)\}$ is divisible by $p$.
    \end{claim}
    \begin{claimproof}
        Let $\emptyset \subsetneq S \subsetneq [m]$.
        We show that
        \begin{equation}
            \label{eq:p-divides-orbit-size}
            |S| \cdot |S^{\Aut(f)}| = m \cdot |\{S' \in S^{\Aut(f)} \mid 1 \in S'\}|.
        \end{equation}
        First observe that this implies the claim, since $|S| < m$ and hence, at least one prime factor of $m = p^k$ needs to divide $|S^{\Aut(f)}|$.

        To show Equation \eqref{eq:p-divides-orbit-size}, consider the $|S^{\Aut(f)}|$ by $m$ matrix $M$ with entries
        \[M_{S',i} = 
        \begin{cases}
            1 &\text{if } i \in S'\\
            0 &\text{otherwise}
        \end{cases}\]
        for all $S' \in S^{\Aut(f)}$ and $i \in [m]$.
        Now, we count the number of ones in the matrix $M$ by rows and by columns.

        First, counting the number of ones by rows gives $|S| \cdot |S^{\Aut(f)}|$ since each row contains exactly $|S|$ many ones since all sets in $S^{\Aut(f)}$ have the same size.
        On the other hand, the first column of $M$ contains $|\{S' \in S^{\Aut(f)} \mid 1 \in S'\}|$ many ones.
        Since $f$ is transitive, all columns contain the same number of ones, so in total we get $m \cdot |\{S' \in S^{\Aut(f)} \mid 1 \in S'\}|$ many one-entries.
        Together, this implies Equation \eqref{eq:p-divides-orbit-size}.
    \end{claimproof}

    We partition the power set $2^{[m]}$ of the set $[m]$ into the orbits $O_0,O_1,\dots,O_\ell$ under the automorphism group $\Aut(f)$.
    Observe that $\{\emptyset\}$ and $\{[m]\}$ are two of the orbits and we may assume without loss of generality that $O_0 = \{\emptyset\}$ and $O_\ell = \{[m]\}$.
    By Claim \ref{claim:p-divides-orbit-size} we get that $p$ divides $|O_j|$ for every $j \in \{1,\dots,\ell-1\}$.
    Let $j \in \{1,\dots,\ell-1\}$.
    By Claim \ref{claim:coefficients-equal-in-orbit} we get that $\altenum{f}(S) = \altenum{f}(S')$ for all $S,S' \in O_j$.
    So
    \[\sum_{S \in O_j} \altenum{f}(S) = 0\]
    since $p$ divides $|O_j|$ and all operations are carried out over $\ZZ_p$.
    Overall, this means that
    \[f(\vec 1) = q_f(\vec 1) = \sum_{S \subseteq [m]} \altenum{f}(S) = \altenum{f}(\emptyset) + \altenum{f}([m]).\]
    Because $f(\vec 0) = q_f(\vec 0) = \altenum{f}(\emptyset)$ and $f(\vec 0) \neq f(\vec 1)$, we conclude that $\altenum{f}([m]) \neq 0$.
\end{proof}

Let $p$ be a prime number and $m \geq 1$.
Let $f\colon \{0,1\}^m \to \ZZ_p$ be a function over $\ZZ_p$ of arity $m$.
Recall that, for $J \subseteq [m]$ and $\ba = (a_j)_{j \in \overline{J}} \in \{0,1\}^{\overline{J}}$, we write $f_{J|\ba}$ to denote the restriction of $f$ to $J$ using $\ba$.

\begin{lemma}
    \label{lem:alternating-enumerator-monotone-prime-power}
    Let $p$ be a prime number and suppose $k = p^\ell$ for some integer $\ell \geq 1$.
    Also suppose $\Phi$ is a graph property on $k$-vertex graphs that is monotone and not $k$-trivial.

    Then there is some $k$-vertex graph $H$ such that $\altenum{\Phi}(H) \neq 0 \bmod p$ and the complete bipartite graph $K_{p^{\ell-1},p^{\ell-1}}$ is a subgraph of $H$.
\end{lemma}

\begin{proof}
    For $i \in \{0,\dots,\ell\}$ we define an equivalence relation $\sim_i$ on $[k]$ via
    \[a \sim_i b \quad\iff\quad \left\lceil\frac{a}{p^i}\right\rceil = \left\lceil\frac{b}{p^i}\right\rceil.\]
    Also, we set $F_i$ to be the graph with vertex set $V(F_i) \coloneqq [k]$ and edge set $E(F_i) \coloneqq \{ab \mid a \neq b \wedge a \sim_i b\}$.
    Note that $F_i$ is a disjoint union of $p^{\ell-i}$ many complete graphs of size $p^i$.
    Moreover,
    \[\emptyset = E(F_0) \subseteq E(F_1) \subseteq \dots \subseteq E(F_\ell) = \binom{[k]}{2}.\]
    Since $\Phi$ is monotone and not $k$-trivial, there is some $i^* \in \{0,\dots,\ell-1\}$ such that $\Phi(F_i) = 1$ for all $i \leq i^*$, and $\Phi(F_i) = 0$ for all $i > i^*$.

    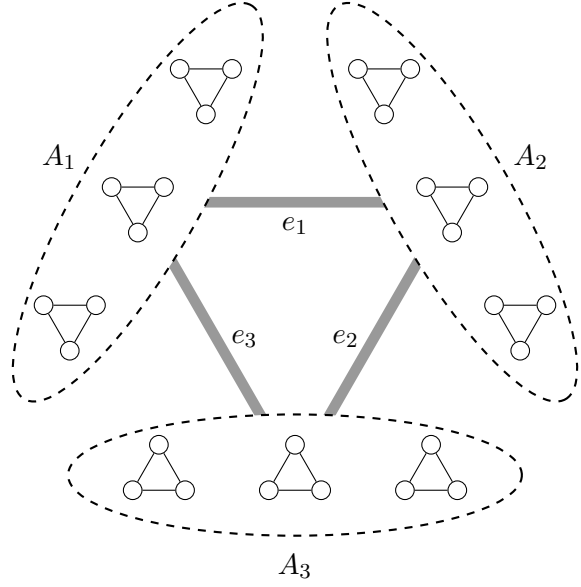
\begin{figure}
        \centering
        \begin{tikzpicture}
            \draw[dashed,thick,rotate around={-60:(30:2.4)},fill=white] (30:2.4) ellipse (3cm and 0.8cm);
            \node at (30:3.6) {$A_2$};
            \foreach[count = \i] \s in {-1,0,1}{
                \foreach[count = \j] \a in {0,120,240}{
                    \node[vertex] (v1-\i-\j) at ($(30:2.4) + (120:1.8*\s) + (30 + \a:0.4)$) {};
                }
                \path (v1-\i-1) edge (v1-\i-2);
                \path (v1-\i-1) edge (v1-\i-3);
                \path (v1-\i-2) edge (v1-\i-3);
            }
            \draw[dashed,thick,rotate around={60:(150:2.4)},fill=white] (150:2.4) ellipse (3cm and 0.8cm);
            \node at (150:3.6) {$A_1$};
            \foreach[count = \i] \s in {-1,0,1}{
                \foreach[count = \j] \a in {0,120,240}{
                    \node[vertex] (v2-\i-\j) at ($(150:2.4) + (240:1.8*\s) + (150 + \a:0.4)$) {};
                }
                \path (v2-\i-1) edge (v2-\i-2);
                \path (v2-\i-1) edge (v2-\i-3);
                \path (v2-\i-2) edge (v2-\i-3);
            }
            \draw[dashed,thick,fill=white] (270:2.4) ellipse (3cm and 0.8cm);
            \node at (270:3.6) {$A_3$};
            \foreach[count = \i] \s in {-1,0,1}{
                \foreach[count = \j] \a in {0,120,240}{
                    \node[vertex] (v3-\i-\j) at ($(270:2.4) + (0:1.8*\s) + (90 + \a:0.4)$) {};
                }
                \path (v3-\i-1) edge (v3-\i-2);
                \path (v3-\i-1) edge (v3-\i-3);
                \path (v3-\i-2) edge (v3-\i-3);
            }

            \scoped[on background layer]{
                \draw[line width=4pt,gray!80] (30:2.4) edge node[below,black] {$e_1$} (150:2.4);
                \draw[line width=4pt,gray!80] (30:2.4) edge node[left,black] {$e_2$} (270:2.4);
                \draw[line width=4pt,gray!80] (150:2.4) edge node[right,black] {$e_3$} (270:2.4);
            }
            
        \end{tikzpicture}
        \caption{Visualization of the construction of the graph $H^*$ in the proof of Lemma \ref{lem:alternating-enumerator-monotone-prime-power} for $p = 3$.
            Each thick gray edge represents a complete biclique between the corresponding sets.}
        \label{fig:levels-monotone}
    \end{figure}

    We define $H_0 \coloneqq F_{i^*}$ and define a second sequence of graphs (see also Figure \ref{fig:levels-monotone}).
    We partition the set $[k]$ into $p$ intervals $A_1,\dots,A_{p}$ of size $p^{\ell-1}$ each, i.e.,
    \[A_i \coloneqq \left\{(i-1) \cdot p^{\ell-1}+1,\dots,i \cdot p^{\ell-1} \right\}.\]
    Also, fix an arbitrary order $\{e_1,\dots,e_m\} = \{vw \mid v \neq w \in [p]\} = \binom{[p]}{2}$ on the two-elements subsets of $[p]$, and define $M_j \coloneqq \{e_1,\dots,e_j\}$ for all $j \in [p]$.
    We define the graph $H_j$ with vertex set $V(H_j) = [k]$ and edge set
    \[E(H_j) \coloneqq E(H_0) \cup \{ab \mid a \in A_v, b \in A_w, vw \in M_j\}.\]
    Observe that
    \[E(H_0) \subseteq E(H_1) \subseteq \dots \subseteq E(H_p).\]
    Also $\Phi(H_0) = 1$ and $\Phi(H_p) = 0$ since $H_p$ contains a subgraph isomorphic to $F_{i^*+1}$.
    So there is some $j^* \in \{0,\dots,p-1\}$ such that $\Phi(H_j) = 1$ for all $j \leq j^*$, and $\Phi(H_j) = 0$ for all $j > j^*$.
    Let us set $H^* \coloneqq H_{j^*}$ and $V \coloneqq A_v$ and $W \coloneqq A_w$ where $vw = e_{j^*}$.

    \begin{claim}
        \label{claim:aut-h-star}
        For every $i \in [p]$ and every $s,s' \in A_i$ there is an automorphism $\sigma \in \Aut(H^*)$ such that $\sigma(s) = s'$ and $\sigma(t) = t$ for all $t \in [k] \setminus A_i$.
    \end{claim}
    \begin{claimproof}
        This follows immediately from observing that $H^*[A_i]$ is a disjoint union of complete graphs of the same size, and every $s,s' \in A_i$ have the same neighbors outside of $A_i$ in $H^*$.
    \end{claimproof}

    Now, consider the polynomial $q_{\Phi}(\bx)$ defined over $\ZZ_p$ where $\bx = (x_e)_{e \in E(K_k)}$, and also let $f_\Phi\colon \{0,1\}^{E(K_k)} \to \ZZ_p$ be the function represented by $q_\Phi$.
    Let 
    \[J \coloneqq \{vw \in E(K_k) \mid v \in V, w \in W\}\]
    and $\overline{J} \coloneqq E(K_k) \setminus J$.
    Let $\ba = (a_e)_{e \in \overline{J}} \in \{0,1\}^{\overline{J}}$ be the tuple defined via
    \[a_e \coloneqq
    \begin{cases}
         1 &\text{if } e \in E(H^*)\\
         0 &\text{otherwise}
    \end{cases}.\]
    Now, consider the function $g = (f_\Phi)_{J|\ba} \colon \{0,1\}^{J} \to \ZZ_p$.
    
    We claim that $g$ satisfies the requirements of Lemma \ref{lem:full-degree-prime-power}.
    First, $|J| = |V \times W| = p^{\ell-1} \cdot p^{\ell-1} = p^{2\ell-2}$ is a prime power of $p$.
    Also, $g(\vec 0) = \Phi(H^*) = 1$ and $g(\vec 1) = \Phi(H_{j^*+1}) = 0$, so $g(\vec 0) \neq g(\vec 1)$.
    It remains to verify that $g$ is transitive.
    
    \begin{claim}
        \label{claim:f-transitive}
        $g$ is transitive.
    \end{claim}
    \begin{claimproof}
        Let $e_1,e_2 \in J$ and suppose $e_1 = s_1t_1$ and $e_2 = s_2t_2$ where $s_1,s_2 \in V$ and $t_1,t_2 \in W$.
        By Claim \ref{claim:aut-h-star} there is automorphism $\sigma \in \Aut(H^*)$ such that $\sigma(s_1) = s_2$ and $\sigma(t_1) = t_2$, and additionally $V = V^{\sigma} \coloneqq \{\sigma(s) \mid s \in V\}$ and $W = W^{\sigma} \coloneqq \{\sigma(t) \mid t \in W\}$.

        Now let $\bar\sigma \in \Sym(E(K_k))$ be defined via $\bar\sigma(st) \coloneqq \sigma(s)\sigma(t)$.
        Observe that $\bar\sigma(e_1) = e_2$ and $J^{\bar \sigma} = J$.
        We claim that $\bar\sigma|_J \in \Aut(g)$.
        Indeed, $\bar \sigma \in \Aut(f_\Phi)$ since $f_\Phi$ represents a graph property.
        Also, $\ba^{\bar\sigma} = (a_{\sigma(e)})_{e \in \overline{J}} = \ba$ since $\sigma \in \Aut(H^*)$.
        Together, this implies that $\bar\sigma|_J \in \Aut(g)$.
    \end{claimproof}
    By Lemma \ref{lem:full-degree-prime-power} we get that $\deg_p(g) = |J|$.
    By Equation \eqref{eq:fourier-restriction} and Lemma \ref{lem:polynomial-alternating-enumerator} this implies that there is some $J \subseteq S \subseteq E(K_k)$ such that $\altenum{\Phi}(K_k[S]) \neq 0 \bmod p$. 
    Hence, we can set $H \coloneqq K_k[S]$.
    Observe that $|V| = |W| = k/p = p^{\ell-1}$ and the complete bipartite graph between $V$ and $W$ is a subgraph of $H$ since $J \subseteq E(H)$.
\end{proof}

\begin{proof}[Proof of Theorem \ref{thm:alternating-enumerator-monotone}]
    Let $\ell$ be maximal such that $p^\ell \leq k$.
    We prove by induction on $d = k - p^\ell$ that there is some $k$-vertex graph $H$ such that $\altenum{\Phi}(H) \neq 0 \bmod p$ and $K_{p^{\ell-1},p^{\ell-1}}$ is a subgraph of $H$.

    For the base case $d = 0$ the statement immediately follows from Lemma \ref{lem:alternating-enumerator-monotone-prime-power}.
    So suppose $d > 0$.
    We distinguish three cases.
    \begin{itemize}
        \item First suppose that $\Phi(K_{1,k-1}) = 1$.
            We define the property $\Psi$ on $(k-1)$-vertex graphs via $\Psi(G) = 1$ if and only if $\Phi(G^+) = 1$ where $G^+$ is the graph obtained from $G$ by adding a universal vertex.
            Then $\Psi$ is a graph property that is monotone and non-trivial on $(k-1)$-vertex graphs.
            Indeed, observe that $\Psi(I_{k-1}) = \Phi(K_{1,k-1}) = 1$, where $I_{k-1}$ is the edge-less graph on $k-1$ vertices, and $\Psi(K_{k-1}) = \Phi(K_k) = 0$.

            By the induction hypothesis, we conclude that there is some $(k-1)$-vertex graph $H'$ such that $\altenum{\Psi}(H') \neq 0 \bmod p$ and $K_{p^{\ell-1},p^{\ell-1}}$ is a subgraph of $H'$.
            
            Now, let $J \coloneqq \{e \in E(K_k) \mid k \notin e\}$.
            Then
            \[f_\Psi = (f_\Phi)|_{J|\vec 1}.\]
            Using Equation \eqref{eq:fourier-restriction} and Lemma \ref{lem:polynomial-alternating-enumerator}, we conclude that there is some $k$-vertex graph $H$ such that $\altenum{\Phi}(H) \neq 0 \bmod p$ and $H' \subseteq H$.
        \item Next suppose that $\Phi(K_{k-1} + K_1) = 0$ where $K_{k-1} + K_1$ is the disjoint union of the complete graphs $K_{k-1}$ and $K_1$.
            We define the property $\Psi$ via $\Psi(G) = 1$ if and only if $\Phi(G + K_1) = 1$ where $G + K_1$ is the graph obtained from $G$ by adding an isolated vertex.
            Then $\Psi$ is a graph property that is monotone and non-trivial on $(k-1)$-vertex graphs.
            Indeed, observe that $\Psi(I_{k-1}) = \Phi(I_k) = 1$ and $\Psi(K_{k-1}) = \Phi(K_{k-1} + K_1) = 0$.

            By the induction hypothesis, we conclude that there is some $(k-1)$-vertex graph $H'$ such that $\widehat{\Psi}(H') \neq 0 \bmod p$ and $K_{p^{\ell-1},p^{\ell-1}}$ is a subgraph of $H'$.
            
            Now, let $J \coloneqq \{e \in E(K_k) \mid k \notin e\}$.
            Then
            \[f_\Psi = (f_\Phi)|_{J|\vec 0}.\]
            Using Equation \eqref{eq:fourier-restriction} and Lemma \ref{lem:polynomial-alternating-enumerator}, we conclude that there is some $k$-vertex graph $H$ such that $\altenum{\Phi}(H) \neq 0 \bmod p$ and $H' \subseteq H$.
        \item Finally, suppose neither of the previous two cases holds.
            Let $r \coloneqq k - 2p^{\ell - 1}$.
            Observe that $r > 0$ since $2p^{\ell - 1} < p^{\ell}$.
            Let $s \coloneqq p^{\ell - 1}$.
            Consider the graph $H^* = K_{r + s} + I_{s}$ on vertex set $[k]$, i.e., $H^*$ is a complete graph on $r+s$ vertices together with $s$ isolated vertices.
            Let $U,V,W$ be a partition of $[k]$ so that $|U| = r$, $|V| = |W| = s$ and $H^*[U \cup V]$ is a complete graph.

            Now, we proceed similarly to the proof of Lemma \ref{lem:alternating-enumerator-monotone-prime-power}.
            Consider the function $f_{\Phi}$ defined over $\ZZ_p$.
            Let 
            \[J \coloneqq \{vw \in E(K_k) \mid v \in V, w \in W\}\]
            and $\overline{J} \coloneqq E(K_k) \setminus J$.
            Let $\ba = (a_e)_{e \in \overline{J}} \in \{0,1\}^{\overline{J}}$ be the tuple defined via
            \[a_e \coloneqq
            \begin{cases}
                 1 &\text{if } e \in E(H^*)\\
                 0 &\text{otherwise}
            \end{cases}.\]
            Now, consider the function $g = (f_\Phi)_{J|\ba} \colon \{0,1\}^{J} \to \ZZ_p$.
    
            We claim that $g$ satisfies the requirements of Lemma \ref{lem:full-degree-prime-power}.
            First, $|J| = p^{\ell-1} \cdot p^{\ell-1} = p^{2\ell-2}$ is a prime power of $p$.
            Also, $g(\vec 1) = 0$ since otherwise $\Phi(K_{1,k-1}) = 1$ by monotonicity which leads to the first case.
            Moreover, $g(\vec 0) = 1$ since otherwise $\Phi(K_{k-1} + K_1) = 0$ by monotonicity which leads to the second case.
            In particular, $g(\vec 0) \neq g(\vec 1)$.
            Also, $g$ is transitive by the same arguments as in Claim \ref{claim:f-transitive}
            
            By Lemma \ref{lem:full-degree-prime-power} we get that $\deg_p(g) = |J|$.
            By Equation \eqref{eq:fourier-restriction} and Lemma \ref{lem:polynomial-alternating-enumerator} this implies that there is some $J \subseteq S \subseteq E(K_k)$ such that $\altenum{\Phi}(K_k[S]) \neq 0 \bmod p$. 
            Hence, we can set $H \coloneqq K_k[S]$.
            Observe that $|V| = |W| = s = p^{\ell-1}$ and the complete bipartite graph between $V$ and $W$ is a subgraph of $H$ since $J \subseteq E(H)$.
    \end{itemize}
    This completes the induction.
    To finish the proof we note that $p^{\ell-1} > \frac{k}{p^2}$.
\end{proof}

By combining Theorem \ref{thm:alternating-enumerator-monotone} for $p = 2$ with Theorems \ref{thm:hardness-w}\ref{item:hardness-w-1} and \ref{thm:hardness-eth}\ref{item:hardness-eth-1}, we obtain the following corollary which improves on the results obtained in \cite{DoringMW24}.

\begin{corollary}
    \label{cor:hardness-monotone}
    There are $N_0,\delta > 0$ such that the following holds.
    Let $k \geq N_0$ and let $\Phi$ be a $k$-vertex graph property that is monotone and not $k$-trivial.
    Then no algorithm solves $\sindsub{\Phi}$ in time $O(n^{\delta \cdot k})$ unless ETH fails.
    
    Moreover, for every computable graph property $\Phi$ that is monotone on $k$-vertex graphs and not $k$-trivial for infinitely many $k$, the problem $\pindsub{\Phi}$ is $\sharpwone$-hard.
\end{corollary}

Also, by combining Theorem \ref{thm:alternating-enumerator-monotone} with Theorems \ref{thm:hardness-w}\ref{item:hardness-w-2} and \ref{thm:hardness-eth}\ref{item:hardness-eth-3}, we obtain hardness for modular counting.

\begin{corollary}
    \label{cor:hardness-monotone-mod}
    Let $p$ be a fixed prime.
    Then there are $N_0,\delta > 0$ such that the following holds.
    Let $k \geq N_0$ and let $\Phi$ be a $k$-vertex graph property that is monotone and not $k$-trivial.
    Then no algorithm solves $\smodindsub{p}{\Phi}$ in time $O(n^{\delta \cdot k})$ unless rETH fails.
    
    Moreover, for every computable graph property $\Phi$ that is monotone on $k$-vertex graphs and not $k$-trivial for infinitely many $k$, the problem $\pmodindsub{p}{\Phi}$ is $\modwone{p}$-hard.
\end{corollary}

Together, Corollary \ref{cor:hardness-monotone} and \ref{cor:hardness-monotone-mod} restate Theorem \ref{thm:intro-hardness-monotone}.

\subsection{Fully Symmetric Properties}

We say that a $k$-vertex graph property $\Phi$ is \emph{fully symmetric} if $\Phi(G) = \Phi(G')$ for all $k$-vertex graphs $G,G'$ with $|E(G)| = |E(G')|$.
In other words, $\Phi$ is fully symmetric if it only depends on the number of edges.
As usual, we say a graph property $\Phi$ is \emph{fully symmetric on $k$-vertex graphs} if $\slice{\Phi}{k}$ is fully symmetric.
After a series of partial results on the hardness of fully symmetric properties \cite{JerrumM15,JerrumM17,RothS20}, it was shown in \cite{RothSW24} that $\sindsub{\Phi}$ is hard for every non-trivial, fully symmetric $k$-vertex graph property.
In the following, we give an alternative proof of hardness that is arguably simpler.

Our proof relies on the following theorem from \cite{GathenR97}.
A Boolean function $f\colon \{0,1\}^m \to \RR$ is \emph{fully symmetric} if $f(\ba) = f(\ba^{\sigma})$ for every $\ba \in \{0,1\}^m$ and every $\sigma \in S_m$, i.e., $f(\ba)$ only depends on the number of $1$-entries in $\ba$.
Observe that a property $\Phi$ on $k$-vertex graphs is fully symmetric if and only if $f_\Phi$ is fully symmetric.

\begin{theorem}[{\cite[Theorem 2.8]{GathenR97}}]
    \label{thm:deg-fully-symmetric}
    Let $f\colon \{0,1\}^m \to \{0,1\}$ be a Boolean function that is fully symmetric and non-constant.
    Then $\deg(f) \geq p - 1$ where $p$ is the largest prime number such that $p \leq m+1$.
\end{theorem}

We note that \cite[Theorem 2.8]{GathenR97} is stated in a slightly different form, but the version stated here immediately follows from the proof.
We also remark that for every $m \geq 2$ there is a prime $p \leq m$ such that $p \geq m - O(m^{0.525})$ using known estimates on the maximal gap between consecutive prime numbers \cite{BakerHP01}.

\begin{theorem}
    \label{thm:alternating-enumerator-symmetric}
    Let $k \geq 2$.
    Suppose $\Phi$ is a $k$-vertex graph property that is fully symmetric and not $k$-trivial.
    Then there is a $k$-vertex graph $H$ such that $\altenum{\Phi}(H) \neq 0$ and $K_\ell$ is a subgraph of $H$ where $\ell \geq 2$ is maximal such that
    \[\binom{\ell}{2}+1 \leq p \leq \binom{k}{2}+1\]
    for some prime $p$.
\end{theorem}

\begin{proof}
    Let $m \coloneqq p-1$ and note that $m \geq 1$.
    Then
    \[\binom{\ell}{2} \leq m \leq \binom{k}{2}.\]
    So there is a set $J \subseteq E(K_k)$ such that $|J| = m$ and $E(K_\ell) \subseteq J$.
    Let us write $\overline{J} \coloneqq E(K_k) \setminus J$.
    
    \begin{claim}
        There is a tuple $\bz \in \{0,1\}^{\overline{J}}$ such that the restriction
        \[g_{\bz} \coloneqq (f_{\Phi})_{J|\bz}\colon \{0,1\}^{J} \to \{0,1\}\]
        of $f_{\Phi}$ to $J$ using $\bz$ is fully symmetric and non-constant.
    \end{claim}
    \begin{claimproof}
        We first show that $g_{\bz} = (f_{\Phi})_{J|\bz}$ is fully symmetric for every tuple $\bz \in \{0,1\}^{\overline{J}}$.
        Indeed, let $\sigma \in \Sym(J)$ be a permutation of $J$, and let $\pi \in \Sym(E(K_k))$ be the unique extension of $\sigma$ that fixes all points outside of $J$.
        Then $g_{\bz}(\ba) = f_\Phi(\ba,\bz) = f_\Phi((\ba,\bz)^{\pi}) = f_\Phi((\ba^{\sigma},\bz)) = g_{\bz}(\ba^{\sigma})$ where the second equality holds since $f_\Phi$ is fully symmetric.

        So it remains to argue that there is some $\bz \in \{0,1\}^{\overline{J}}$ such that $g_{\bz}$ is non-constant.
        Suppose towards a contradiction that $g_{\bz}$ is constant for every $\bz \in \{0,1\}^{\overline{J}}$.
        This means that for every $\bz \in \{0,1\}^{\overline{J}}$ there is some $\alpha(\bz) \in \{0,1\}$ such that $g_{\bz}$ always evaluates to $\alpha(\bz)$.
        Since $f_\Phi$ is non-constant, there are $\bz_1,\bz_2 \in \{0,1\}^{\overline{J}}$ such that $\alpha(\bz_1) \neq \alpha(\bz_2)$.
        Additionally, we may assume that $\bz_1$ and $\bz_2$ only differ in one entry.
        But now there are $\by_1,\by_2 \in \{0,1\}^{J}$ such that $(\by_1,\bz_1)$ and $(\by_2,\bz_2)$ have the same number of $1$-entries since $m \geq 1$.
        So $f_\Phi(\by_1,\bz_1) = \alpha(\bz_1) \neq \alpha(\bz_2) = f_\Phi(\by_2,\bz_2)$.
        But this is a contradiction since $f_\Phi$ is fully symmetric.
    \end{claimproof}
    
    So $\deg(g_{\bz}) = m$ by Theorem \ref{thm:deg-fully-symmetric} which implies that $\altenum{g_{\bz}}(J) \neq 0$.
    By Equation \eqref{eq:fourier-restriction} there is some set $T \subseteq \overline{J}$ such that $\altenum{f_\Phi}(J \cup T) \neq 0$.
    So $\altenum{\Phi}(K_k[J \cup T]) \neq 0$ by Lemma \ref{lem:polynomial-alternating-enumerator}.
    To finish the proof observe that $K_\ell$ is a subgraph of $H \coloneqq K_k[J \cup T]$ since $E(K_\ell) \subseteq J$.
\end{proof}

Note that, by Bertrand's Postulate, it is always possible to choose $\ell \geq \frac{k}{2}$.
In fact, using the estimate on the maximal gap between consecutive prime numbers \cite{BakerHP01} discussed above, a simple calculation shows that it is possible to choose $\ell \geq k - O(k^{0.05})$.
So by combining Theorem \ref{thm:alternating-enumerator-symmetric} with Theorems \ref{thm:hardness-w}\ref{item:hardness-w-1} and \ref{thm:hardness-eth}\ref{item:hardness-eth-1}, we obtain the following.

\begin{corollary}[Theorem \ref{thm:intro-hardness-symmetric} restated]
    \label{cor:hardness-symmetric}
    There are $N_0,\delta > 0$ such that the following holds.
    Let $k \geq N_0$ and let $\Phi$ be a $k$-vertex graph property that is fully symmetric and not $k$-trivial.
    Then no algorithm solves $\sindsub{\Phi}$ in time $O(n^{\delta \cdot k})$ unless ETH fails.
    
    Moreover, for every computable graph property $\Phi$ that is fully symmetric on $k$-vertex graphs and not $k$-trivial for infinitely many $k$, the problem $\pindsub{\Phi}$ is $\sharpwone$-hard.
\end{corollary}

\begin{remark}
    We note that previous techniques by Roth et al.~\cite{RothSW24} only rule out an algorithm solving $\sindsub{\Phi}$ in time $O(n^{\delta \cdot k/\sqrt{\log k} })$ assuming ETH.
    However, combining the previous techniques with Theorem \ref{fact:density-makes-hard}, which was shown only very recently~\cite{CurticapeanDNW25}, a lower bound of $O(n^{\delta \cdot k})$ under ETH would also follow.

    On the other hand, our arguments remain valid even without the tighter lower bounds provided by Theorem \ref{fact:density-makes-hard}, since we obtain a $k$-vertex graph $H$ with $\altenum{\Phi}(H) \neq 0$ such that $H$ is not only dense, but it even contains an $\ell$-clique for $\ell \geq k/2$.
    Under different complexity assumptions than ETH, this can yield stronger lower bounds; see Section~\ref{sec:wl}.
\end{remark}

\subsection{Invariants with Few Distinct Hamming Weights}

In this section, we use the polynomial representation of graph invariants to derive lower bounds for $\pindsub{\Phi}$ when $\Phi$ avoids a superlinear number of Hamming weights.
In particular, our theorem allows us to recover \cite[Theorem~3.4]{JerrumM15} on hardness for $\Phi$ with $o(k^2)$ distinct Hamming weights, and the more general \cite[Main Theorem 5]{RothSW24}, which applies whenever $\omega(k)$ Hamming weights are avoided.
Our result even holds for graph invariants rather than properties, provided that they map into a field of characteristic zero.

Given a $k$-vertex graph invariant $\Phi$, define
\[\hw(\Phi) \coloneqq \{|E(H)| \mid \Phi(H)\neq 0\}\]
to be the set of all edge-counts (i.e, Hamming weights, when viewed as bitstrings) among the support of $\Phi$.
Throughout this section, let $m \coloneqq \binom{k}{2}$.
We define a value that counts the avoided Hamming weights, 
except when $\Phi$ is empty, as
\[\beta_\Phi \coloneqq
    \begin{cases}
        0 & \text{if } |\hw(\Phi)| = 0,\\
        m + 1 - |\hw(\Phi)| & \text{otherwise}.
    \end{cases}
\]
For example, if $\Phi$ is trivial, then $\beta_\Phi = 0$. If $\Phi$ is the planarity property of $k$-vertex graphs, then $\beta_\Phi = m - 3k+6$.
For a general graph invariant $\Phi$ and $k\in\NN$, we define $\beta_\Phi(k) \coloneqq \beta_{\slice{\Phi}{k}}$.

\begin{theorem}
    \label{thm:alternating-enumerator-hw}
    Let $\Phi$ be a $k$-vertex graph invariant over a field $\FF$ of characteristic $0$ that is not $k$-trivial.
    Then there is a $k$-vertex graph $H$ with $\altenum{\Phi}(H) \neq 0$ and $|E(H)| \geq \beta_\Phi/2$.
    
    If additionally $\FF = \QQ$ and $\Phi(F) \geq 0$ for all $k$-vertex graphs $F$, then $|E(H)| \geq \beta_\Phi$.
\end{theorem}

Let $\Phi$ be a $k$-vertex graph invariant over a field $\FF$.
To prepare the proof of our lower bound, we define an invariant $\overline \Phi$ from $\Phi$ by setting $\overline \Phi(H) \coloneqq \Phi(\overline H)$, where $\overline H$ is the complement graph of $H$.

\begin{fact}
    \label{fact:deg-complement}
    If $\Phi$ be a $k$-vertex graph invariant over a field $\FF$, then $\deg_\FF(q_\Phi) = \deg_\FF(q_{\overline \Phi})$.
\end{fact}

\begin{proof}
    Let $q_\Phi(\bx)$ and $q_{\overline{\Phi}}(\bx)$ denote the polynomial representations of $\Phi$ and $\overline{\Phi}$, respectively.
    From the definition, we directly obtain that $q_\Phi(x_1,\ldots,x_m) = q_{\overline \Phi}(1-x_1,\ldots,1-x_m)$, where $\bx = (x_1,\dots,x_m)$.
    It follows that $\deg_\FF(q_\Phi) = \deg_\FF(q_{\overline \Phi})$.
\end{proof}

The following simple fact will also be useful:

\begin{fact}
    \label{fact:beta-phi-complement}
    If $\Phi$ be a $k$-vertex graph invariant over a field $\FF$ which is not $k$-trivial, then there is a graph $F$ with $|E(F)|\leq \beta_\Phi/2$ such that $\Phi(F) \neq 0$ or $\overline\Phi(F) \neq 0$.
\end{fact}

\begin{proof}
    The invariant $\Phi$ cannot avoid all of the weights $0,\dots,\lfloor\beta_\Phi/2 \rfloor$ and $m,m-1,\dots,m-\lfloor\beta_\Phi/2 \rfloor$ (since those are at least $\beta_\Phi+1$ weights). 
    In the first case, there is a graph $F$ with $|E(F)|\leq \beta_\Phi/2$ such that $\Phi(F) \neq 0$.
    Otherwise, there is a graph $H$ with $|E(H)| \geq m - \beta_\Phi/2$ such that $\Phi(H) \neq 0$.
    We set $F \coloneqq \overline{H}$ which satisfies $|E(F)|\leq \beta_\Phi/2$ and $\overline\Phi(F) \neq 0$.
\end{proof}

We are ready for the main proof.

\begin{proof}[Proof of Theorem~\ref{thm:alternating-enumerator-hw}]
    We prove the first statement and then describe the modifications required for the second statement.
    Let $q_\Phi(\bx)$ be the polynomial representation of $\Phi$ in variables $\bx = (x_1,\ldots,x_m)$; slightly abusing notation, we identify the variables $(x_1,\dots,x_m)$ with $(x_e)_{e \in E(K_k)}$ by fixing an arbitrary numbering of the edges of $K_k$.
    By Corollary~\ref{cor:alt-enum-vs-degree} and Fact~\ref{fact:deg-complement}, it suffices to show that one of $q_\Phi$ or $q_{\overline \Phi}$ has degree at least $\beta_\Phi /2$.
    
    Let $F$ be a graph with minimal number of edges $r \coloneqq |E(F)|$ such that $\Phi(F)\neq 0$.
    We may assume without loss of generality that $r \leq \beta_\Phi /2$, since otherwise we can replace $\Phi$ by $\overline{\Phi}$ in the remainder of the argument, by Fact~\ref{fact:beta-phi-complement}.

    Let $S \coloneqq E(F)$.
    We define the polynomial
    \[q(\bx) \coloneqq q_\Phi(\bx) \cdot x^S ~=~ q_\Phi(\bx) \cdot \prod_{e \in S}x_e.\]
    Clearly $\deg(q) \leq \deg(q_\Phi) + r$.
    On every input $\ba \in \{0,1\}^m$ of Hamming weight~$w \notin \hw(\Phi)$, we have $q(\ba) = 0$, since then $q_\Phi(\ba)=\Phi(K_k[\ba]) = 0$.
    Moreover, if $\ba \in \{0,1\}^m$ has Hamming weight~$r$, we have
    \begin{equation}
    \label{eq:filterphi}
        q(\ba) = \begin{cases}
            \Phi(F) & \text{if }K_k[\ba] = F, 
            \\ 0 & \text{otherwise.}
        \end{cases}
    \end{equation}
    
    We show $\deg(q) \geq \beta_\Phi$, which implies $\deg(q_\Phi) \geq \beta_\Phi -r \geq \beta_\Phi / 2$ and proves the lemma.
    Towards this end, we define a symmetric polynomial from $q$ via
    \[q^\mathrm{sym}(x_1,\dots,x_m) = \sum_{\pi \in S_m} q(x_{\pi(1)},\dots,x_{\pi(m)}).\]
    Clearly, $\deg(q^\mathrm{sym}) \leq \deg(q)$.
    Moreover, $q^\mathrm{sym}$ can be ``compressed'' to a univariate polynomial: By~\cite[Theorem~2.2]{Jukna12} (see also~\cite{NisanS94,MinskyP87}), there exists a univariate polynomial $\tilde q$ with $\deg(\tilde q) \leq \deg(q^\mathrm{sym})$ and $\tilde q(\xi_1 + \dots + \xi_m) = q^\mathrm{sym}(\xi_1,\ldots,\xi_m)$ for all $\xi_1,\ldots,\xi_m \in \{0,1\}$.
    \begin{claim}
        \label{claim:deg-f-beta-Phi}
        $\deg(\tilde q) \geq \beta_\Phi$.
    \end{claim}
    \begin{claimproof}
    The polynomial $\tilde q(x)$ has at least $\beta_\Phi$ roots, since $\tilde q(b) = 0$ for all $b \in \{0,\dots,m\} \setminus \hw(\Phi)$.
    This implies $\deg(\tilde q) \geq \beta_\Phi$ if $\tilde q \not \equiv 0$.
    To show $\tilde q \not \equiv 0$, observe that only $\Phi(F)$ contributes to $\tilde q(r)$, that is,  
    \[\tilde q(r) ~=~ r! (m-r)!\sum_{\substack{\ba \in \{0,1\}^m \text{ of}\\\text{Hamming weight }r} } q(\ba) ~=~ r! (m-r)! \cdot \Phi(F) ~\neq~ 0,\]
    where the second equality follows from~\eqref{eq:filterphi}.
    Since $\tilde q(r)\neq 0$ (where we use that $\FF$ has characteristic $0$) we indeed have $\tilde q \not \equiv 0$.
    \end{claimproof}
    
    The lemma now follows directly from Corollary~\ref{cor:alt-enum-vs-degree} and the inequalities
    \[
    \deg(q_\Phi) + \beta_\Phi / 2 
    ~\geq~ \deg(q)
    ~\geq~ \deg(q^\mathrm{sym})
    ~\geq~ \deg(\tilde q)
    ~\geq~ \beta_\Phi.
    \]

    For the second claim, we set $q(\bx) = q_\Phi(\bx)$, so $\deg(q)=\deg(q_\Phi)$ holds without the additive term $r$. We proceed as before, ignoring \eqref{eq:filterphi}, and show $\tilde q \not \equiv 0$ in Claim~\ref{claim:deg-f-beta-Phi} by observing that $\tilde q(r)=r! (m-r)!\sum_{\ba} q(\ba) >0$, since all terms in the sum are nonnegative and (at least) the terms corresponding to $F$ are positive.
\end{proof}

By combining Theorem~\ref{thm:alternating-enumerator-hw} with Theorems \ref{thm:hardness-w}\ref{item:hardness-w-1} and \ref{thm:hardness-eth}\ref{item:hardness-eth-1}, we obtain the following.
For functions $f,g\colon \NN \to \QQ$, we say that $f \in \omega_{\mathrm{inf}}(g)$ if there is an infinite sequence $n_0,n_1,n_2,\ldots$ such that $f(n_i) \in \omega(g(n_i))$. Note that $f(n_i)$ and $g(n_i)$ take $i$ as input.
Recall that $\beta_\Phi(k) = \beta_{\slice{\Phi}{k}}$, where $\slice{\Phi}{k}$ is the $k$-th slice of $\Phi$.

\begin{corollary}[Theorem \ref{thm:intro-hardness-hw} restated]
    \label{cor:hardness-hw}
    There are $N_0,\delta > 0$ such that the following holds.
    Let $k \geq N_0$ and $0 < d \leq k/2$, and let $\Phi$ be a $k$-vertex graph invariant such that
    \[\beta_\Phi \geq d \cdot k.\]
    Then no algorithm solves $\sindsub{\Phi}$ in time $O(n^{\delta \cdot d})$ unless ETH fails.
    
    Moreover, for every computable graph invariant $\Phi$ such that
    \[\beta_\Phi(k) \in \omega_{\mathrm{inf}}(k),\] the problem $\pindsub{\Phi}$ is $\sharpwone$-hard.
\end{corollary}

We note that the lower bound in Corollary \ref{cor:hardness-hw} is essentially tight:
For every $4 \leq d \leq k/2$, there is a $k$-vertex property $\Phi$ with $\beta_\Phi \approx d \cdot k$ such that $\sindsub{\Phi}$ can be solved in time $O(n^d)$; see \cite[Remark 3.6]{CurticapeanDN25}.

\subsection{Invariants with Small Image}
\label{sec:hardness-small}

We also derive hardness results for graph invariants with sufficiently small image.
Let $\Phi$ be a $k$-vertex graph invariant over $\FF$. We write $\im(\Phi)$ for its image and define
\[\iota_\Phi \coloneqq |\im(\Phi)| - 1.\]
For a general graph invariant $\Phi$, we define $\iota_\Phi(k) \coloneqq \iota_{\slice{\Phi}{k}}$ for $k\in\NN$.

Let $\tau\colon \im(\Phi) \to \FF$ be a function that, intuitively speaking, relabels the outputs of $\Phi$.
Formally, we define the graph invariant $\Phi \circ \tau$ via
\[(\Phi \circ \tau)(G) \coloneqq \tau(\Phi(G)).\]

\begin{lemma}
    \label{lem:alternating-enumerator-small-image}
    Let $\Phi$ be a $k$-vertex graph invariant over $\FF$ that is not $k$-trivial and let $\tau\colon \im(\Phi) \to \FF$.
    Define $\Psi \coloneqq \Phi \circ \tau$ and let $F$ be a $k$-vertex graph such that $\altenum{\Psi}(F) \neq 0$.
    Then there is a $k$-vertex graph $H$ with $\altenum{\Phi}(H) \neq 0$ and
    \[|E(H)| \geq \frac{|E(F)|}{\iota_\Phi}.\]
\end{lemma}

\begin{proof}
    Let $m \coloneqq \binom{k}{2}$.
    Let $q_\Phi(\bx)$ and $q_\Psi(\bx)$ be the polynomial representations of $\Phi$ and $\Psi$.

    \begin{claim}
        \label{claim:look-up-table}
        $\deg(q_\Psi) \leq \deg(q_\Phi) \cdot \iota_\Phi$.
    \end{claim}
    \begin{claimproof}
        Using interpolation, there is a univariate polynomial $p(x)$ of degree at most $\iota_\Phi$ such that $p(a) = \tau(a)$ for all $a \in \im(\Phi)$.
        Now, consider the polynomial $q_\Psi'(\bx) = p(q_\Phi(\bx))$.
        Clearly, $\deg(q_\Psi') \leq \deg(q_\Phi) \cdot \iota_\Phi$.
        Also, $q_\Psi'$ represents $\Psi$ over $\{0,1\}^m$.

        Given a polynomial $r(\bx)$, let $r_{\textsf{red}}$ denote the polynomial obtained by replacing $x_i^d$ with $x_i$ in every monomial.
        We clearly have $r(\ba) = r_{\textsf{red}}(\ba)$ for all $\ba \in \{0,1\}^m$.
        Since every multilinear polynomial in $m$ indeterminates is determined by its evaluations on all values $\ba \in\{0,1\}^m$, we obtain $q_\Psi = (q_\Psi')_{\textsf{red}}$.
        In particular, $\deg(q_\Psi) \leq \deg(q_\Psi')$.
    \end{claimproof}

    Now, the lemma immediately follows from Corollary~\ref{cor:alt-enum-vs-degree} and Claim~\ref{claim:look-up-table}.
    Indeed, $|E(F)| \leq \deg(q_\Psi)$ by Corollary~\ref{cor:alt-enum-vs-degree}.
    So $\deg(q_\Phi) \geq |E(F)|/\iota_\Phi$ by Claim~\ref{claim:look-up-table}.
    Using Corollary~\ref{cor:alt-enum-vs-degree} again, we obtain a $k$-vertex graph $H$ with $\altenum{\Phi}(H) \neq 0$ and $|E(H)| \geq \deg(q_\Phi) \geq |E(F)|/\iota_\Phi$.
\end{proof}

To demonstrate the applicability of Lemma \ref{lem:alternating-enumerator-small-image}, let us consider the example of monotone graph invariants with small images.
This setting has first been studied in \cite{DoringMW25}, and we can use Lemma \ref{lem:alternating-enumerator-small-image} to obtain improved lower bounds.
Let $k \geq 1$ and let $\Phi$ be a $k$-vertex graph invariant.
We say that $\Phi$ is \emph{monotone} if $\Phi(G') \geq \Phi(G)$ for every $k$-vertex graph $G$ and every edge-induced subgraph $G'$ of $G$.
For an arbitrary invariant $\Phi$, we say that $\Phi$ is \emph{monotone on $k$-vertex graphs} if $\slice{\Phi}{k}$ is monotone.

\begin{corollary}
    \label{cor:alternating-enumerator-monotone-small-image}
    Let $\Phi$ be a $k$-vertex graph invariant that is monotone and not $k$-trivial.
    Then there is some $k$-vertex graph $H$ such that $\altenum{\Phi}(H) \neq 0$ and
    \[|E(H)| \geq \frac{k^2}{16 \cdot\iota_\Phi}.\]
\end{corollary}

\begin{proof}
    Let $I_k$ denote the edge-less $k$-vertex graph and let $a \coloneqq \Phi(I_k)$.
    Let $\tau\colon \im(\Phi) \to \QQ$ be defined via $\tau(a) = 1$ and $\tau(b) = 0$ for all $a \neq b \in \im(\Phi)$.
    Then $\Psi \coloneqq \Phi \circ \tau$ is a $k$-vertex graph property that is monotone and not $k$-trivial.
    By Theorem \ref{thm:alternating-enumerator-monotone} there is a $k$-vertex graph $F$ such that $\altenum{\Psi}(F) \neq 0$ and $|E(F)| \geq k^2/16$.
    So, by Lemma \ref{lem:alternating-enumerator-small-image}, there is a $k$-vertex graph $H$ such that $\altenum{\Phi}(H) \neq 0$ and
    \[|E(H)| \geq \frac{|E(F)|}{\iota_\Phi} \geq \frac{k^2}{16 \cdot \iota_\Phi}.\qedhere\]
\end{proof}

By combining Corollary \ref{cor:alternating-enumerator-monotone-small-image} with Theorems \ref{thm:hardness-w}\ref{item:hardness-w-1} and \ref{thm:hardness-eth}\ref{item:hardness-eth-1}, we obtain the following result.
We say a graph invariant $\Phi$ is \emph{infinitely often monotone with sublinear image} if for every $\varepsilon > 0$ there is some $k \in \NN$ such that $\slice{\Phi}{k}$ is monotone and not $k$-trivial, and $\iota_{\Phi}(k) \leq \varepsilon \cdot k$.

\begin{corollary}
    \label{cor:hardness-monotone-small-image}
    There are $N_0,\delta > 0$ such that the following holds.
    Let $k \geq N_0$ and let $\Phi$ be a $k$-vertex graph invariant that is monotone and not $k$-trivial.
    Then no algorithm solves $\sindsub{\Phi}$ in time $O(n^{\delta \cdot k/\iota_\Phi})$ unless ETH fails.

    Moreover, for every computable graph invariant $\Phi$ that is infinitely often monotone with sublinear image the problem $\pindsub{\Phi}$ is $\sharpwone$-hard.
\end{corollary}

We note that Corollary~\ref{cor:hardness-monotone-small-image} improves over \cite[Main Theorem 2]{DoringMW25} which only obtains $\sharpwone$-hardness if $\iota_\Phi(k) \leq (1 - \varepsilon) \cdot \sqrt{k}$.

\section{All Degrees are Even}
\label{sec:even}

In this section we consider the \emph{all-even} property $\alleven$ defined via
\[\alleven(G) = \begin{cases}
                    1 &\text{if } \deg_G(v) \equiv 0 \bmod 2 \text{ for all } v \in V(G)\\
                    0 &\text{otherwise,}
                \end{cases}
\]
i.e., $\alleven$ contains exactly those graphs where all vertices have even degree.
This property turns out to be interesting since $\altalleven(H) = 0$ for most graphs $H$.
While $\pindsub{\alleven}$ is still $\sharpwone$-hard, this property allows us to also derive non-trivial upper bounds on the complexity.
Moreover, the fact that $\altalleven(H) = 0$ for most graphs $H$ surprisingly allows us to obtain hardness of $\pindsub{\Phi}$ for other properties $\Phi$.

\subsection{Complexity of the All-Even Property}

The starting point for our analysis is the observation that each slice $\slice{\alleven}{k}$ forms a vector space over the $2$-element field $\ZZ_2$.
Indeed, $\slice{\alleven}{k}$ contains exactly those graphs that are in the \emph{cycle space} of $K_k$.
This allows us to use basic tools from linear algebra to determine exactly which graphs $H$ satisfy $\altalleven(H) \neq 0$.

\begin{theorem}
    \label{thm:alternating-enumerator-alleven}
    The graphs $H$ satisfying $\altalleven(H) \neq 0$ are precisely the bipartite graphs.
\end{theorem}

\begin{proof}
    Let $H$ be a bipartite graph.
    For every $S \subseteq E(H)$ with $\alleven(H[S]) = 1$ we have $|S|$ even, so every term in
    \[(-1)^{|E(H)|} \cdot \altalleven(H) = \sum_{S \subseteq E(H)} (-1)^{|S|} \cdot \alleven(H[S])\]
    is non-negative.
    Since the term $\alleven(H[\emptyset]) = 1$ is strictly positive, we get $\altalleven(H) \neq 0$.

    Next, let $H$ be a non-bipartite graph and let $C^* \subseteq E(H)$ be the edges of an odd cycle in $H$.
    Consider the cycle space of $H$, i.e., the set of all subgraphs in which all vertices have even degree: 
    \[\CC \coloneqq \{S \subseteq E(H) \mid H[S] \in \alleven\}.\]
    As noted above, $\CC$ can be viewed as a subspace of $\ZZ_2^{E(H)}$ where each $C \in \CC$ is identified with its corresponding indicator vector $v_C \in \ZZ_2^{E(H)}$.
    Slightly abusing notation, we identify $C$ with $v_C$. 

    \begin{claim}
        There is a basis $\{C_1,\dots,C_d\} \subseteq \CC$ such that $|C_i|$ is odd for every $i \in [d]$.
    \end{claim}
    \begin{claimproof}
        Clearly, there is a basis $\CB$ of $\CC$ such that $C^* \in \CB$.
        But then
        \[\CB^* \coloneqq \{C \mid C \in \CB, |C| \text{ is odd}\} \cup \{C \oplus C^* \mid C \in \CB, |C| \text{ is even}\}\]
        is also a basis of $\CC$, where $C \oplus C^*$ denotes the symmetric difference of $C$ and $C^*$ (i.e., the corresponding vectors over $\ZZ_2$ are added).
        Every element of $\CB^*$ has odd cardinality, as desired.
    \end{claimproof}

    Now let $\{C_1,\dots,C_d\}$ be the basis constructed in the claim.
    Since $|C_i|$ is odd for every $i \in [d]$, we get that
    \[|\bigoplus_{i \in M} C_i| \equiv |M| \mod 2\]
    for every $M \subseteq [d]$.
    It follows that
    \begin{align*}
        \altalleven(H) &= (-1)^{|E(H)|} \cdot \sum_{S \subseteq E(H)} (-1)^{|S|} \cdot \alleven(H[S])\\
                       &= (-1)^{|E(H)|} \cdot \sum_{C \in \CC} (-1)^{|C|}\\
                       &= (-1)^{|E(H)|} \cdot \sum_{M \subseteq [d]} (-1)^{|\bigoplus_{i \in M} C_i|}\\
                       &= (-1)^{|E(H)|} \cdot \sum_{M \subseteq [d]} (-1)^{|M|}\\
                       &= 0
    \end{align*}
    where the last equality follows from the Binomial Theorem.
\end{proof}

\begin{remark}
    \label{rem:alleven}
    A weaker version of Theorem \ref{thm:alternating-enumerator-alleven} can also be obtained via the Fourier expansion.
    Let $k \geq 1$ and let $\Psi \coloneqq \slice{\alleven}{k}$.
    The set $\supp(f_{\Psi})$ forms a subspace of $\ZZ_2^{m}$ for $m \coloneqq \binom{k}{2}$.
    In this case, the support $\supp(\widehat{f}_{\Psi})$ in the Fourier expansion is known to be the complement space; see \cite[Proposition 3.11]{ODonnell14}.
    Since $\supp(f_{\Psi})$ corresponds to the \emph{cycle space} of $K_k$, we obtain that $\supp(\widehat{f}_{\Psi})$ corresponds to the \emph{cut space} of $K_k$, i.e., the complete bipartite graphs on $k$ vertices.
    Using \eqref{eq:poly-basis-change}, it is possible to conclude that
    \begin{enumerate}
        \item $\altalleven(H) \neq 0$ for all complete bipartite graphs $H$, and
        \item if $\altalleven(H) \neq 0$ then $H$ is bipartite.
    \end{enumerate}
    Our Theorem~\ref{thm:alternating-enumerator-alleven} shows that \emph{all} bipartite graphs $H$ are in the support of $\altalleven(H)$, not just the complete bipartite graphs.
\end{remark}

By combining Theorem \ref{thm:alternating-enumerator-symmetric} with Theorems \ref{thm:hardness-w}\ref{item:hardness-w-1} and \ref{thm:hardness-eth}\ref{item:hardness-eth-1}, we obtain the following.

\begin{corollary}
    \label{cor:hardness-alleven}
    There are $N_0,\delta > 0$ such that for every $k \geq N_0$ there is no algorithm that solves $\sindsub{\slice{\alleven}{k}}$ in time $O(n^{\delta \cdot k})$ unless ETH fails.
    Moreover, $\pindsub{\alleven}$ is $\sharpwone$-hard.
\end{corollary}

In contrast to the results in Section \ref{sec:hardness-fourier}, Theorem \ref{thm:alternating-enumerator-alleven} also allows us to obtain a non-trivial upper bound on the complexity of $\pindsub{\alleven}$.
This follows from Theorem \ref{thm:algorithm-alternating-enumerator-vc}, by observing that each bipartite graph with $k$ vertices has vertex cover number at most $\lfloor k/2\rfloor$.

\begin{corollary}
    The problem $\pindsub{\alleven}$ can be solved in time $f(k) \cdot n^{\lfloor k/2\rfloor + 1}$ for some computable function $f$.
\end{corollary}

Observe that the algorithm from the last corollary achieves roughly a quadratic speed-up over the brute-force algorithm, which runs in time $f(k) \cdot n^{k}$ for some computable function $f$.

\subsection{Showing Hardness via the All-Even Property}

Next, we use Theorem \ref{thm:alternating-enumerator-alleven} to show hardness of $\sindsub{\Phi}$ for other graph invariants.
Indeed, Theorem \ref{thm:alternating-enumerator-alleven} shows that the all-even property is ``relatively easy'' in the sense that only few graphs have non-zero alternating enumerator.
Surprisingly, this fact allows us to prove hardness results for other graph invariants.

More precisely, to obtain hardness results for $\sindsub{\Phi}$ for a $k$-vertex graph invariant $\Phi$, the key idea is to take a point-wise product with the all-even property.
Formally, the \emph{point-wise product} of two graph invariants $\Phi,\Gamma$ is the graph invariant $\Psi = \Phi \cdot \Gamma$ defined via $\Psi(G) = \Phi(G) \cdot \Gamma(G)$.

In our applications, consider a $k$-vertex graph invariant $\Phi$ and the pointwise product $\Psi \coloneqq \Phi \cdot \alleven$.
If $\altenum{\Psi}(H^*) \neq 0$ for some $k$-vertex graph $H^*$ that is ``far away'' from being bipartite, we can show that there is some $k$-vertex graph $H$ with many edges such that $\altenum{\Phi}(H) \neq 0$ (see Lemma \ref{lem:alternating-enumerator-intersection}).
This allows us to lift hardness results from $\Psi$ to the invariant $\Phi$.

In the following, we demonstrate this approach on two examples.
In all cases, the starting point is the following lemma that expresses an alternating enumerator for a pointwise product via the alternating enumerators of its factors.

\begin{lemma}
    \label{lem:alternating-enumerator-intersection}
    Let $\Phi,\Gamma$ be $k$-vertex graph invariants and $\Psi \coloneqq \Phi \cdot \Gamma$.
    Then, for every graph $H$,
    \[\altenum{\Psi}(H) = \sum_{\substack{S_1,S_2 \subseteq E(H)\colon\\ S_1 \cup S_2 = E(H)}} \altenum{\Phi}(H[S_1]) \cdot \altenum{\Gamma}(H[S_2]).\]
\end{lemma}

\begin{proof}
    Let $q_{\Phi}(\bx)$, $q_{\Gamma}(\bx)$, and $q_{\Psi}(\bx)$ be the polynomial representations of $\Phi,\Gamma$ and $\Psi$, respectively.
    Given a polynomial $r(\bx)$, let $r_{\textsf{red}}(\bx)$ denote the multilinear polynomial obtained by replacing $x_i^d$ with $x_i$ in every monomial.
    We clearly have $r(\ba) = r_{\textsf{red}}(\ba)$ for all $\ba \in \{0,1\}^m$.
    Since every multilinear polynomial in $m$ indeterminates is determined by its evaluations on all values $\ba \in\{0,1\}^m$, we obtain $q_\Psi(\bx) = (q_\Phi(\bx) \cdot q_\Gamma(\bx))_{\textsf{red}}$, since $q_\Psi(\ba) = q_\Phi(\ba) \cdot q_\Gamma(\ba)$ for all $\ba \in \{0,1\}^m$.
    We now consider the coefficients of $q_\Psi(\bx)$.

    Let $K= K_k$.
    In the product $q_\Psi(\bx) = q_\Phi(\bx) \cdot q_\Gamma(\bx)$, monomials $x^S$ represent multisets $S$ of edges from $E(H)$ with multiplicity at most $2$ per edge.
    By Lemma~\ref{lem:polynomial-alternating-enumerator}, the coefficients of $x^S$ in $q_\Phi(\bx)$ and $q_\Gamma(\bx)$ are $\altenum{\Phi}(K[S])$ and $\altenum{\Gamma}(K[S])$, respectively.
    By collecting monomials after polynomial multiplication, the coefficient of $x^S$ in $q_\Phi(\bx) \cdot q_\Gamma(\bx)$ equals
    \begin{equation}
        \label{eq:convolution}
        \sum_{\substack{S_1,S_2 \subseteq S \colon\\ S_1 \cup S_2 = S}} \altenum{\Phi}(K[S_1]) \cdot \altenum{\Gamma}(K[S_2]),
    \end{equation}
    where the union $S_1 \cup S_2$ is interpreted as a multiset union, i.e., counting copies.
    Collecting monomials after substitution, it follows that the coefficient $c_S$ of $x^S$ in $q_\Psi(\bx)=(q_\Phi(\bx) \cdot q_\Gamma(\bx))_{\textsf{red}}$ for a \emph{set} (rather than multiset) $S \subseteq E(H)$ also equals 
    \eqref{eq:convolution}, but with $S_1 \cup S_2$ interpreted as a set union, i.e., not counting copies.
\end{proof}

As the first example, we consider properties $\Phi$ whose intersection with $\alleven$ is non-empty and contains only bipartite graphs.
In this situation, setting $\Psi \coloneqq \Phi \cdot \alleven$, we can show that $\altenum{\Psi}(K_k) \neq 0$.
Using Lemma \ref{lem:alternating-enumerator-intersection} and Theorem \ref{thm:alternating-enumerator-alleven}, we can argue that $\altenum{\Phi}(H) \neq 0$ for some graph $H$ with $\Omega(k^2)$ edges.

\begin{theorem}
    \label{thm:alternating-enumerator-alleven-bipartite}
    For every $k$-vertex graph property $\Phi$ whose intersection with $\alleven$ is non-empty and contains only bipartite graphs, there is a $k$-vertex graph $H$ such that $\altenum{\Phi}(H) \neq 0$ and $K_\ell$ is a subgraph of $H$ for some $\ell \geq \frac{k}{2}$.
\end{theorem}

\begin{proof}
    Let $\Psi \coloneqq \Phi \cdot \alleven$.
    Note that $\Psi$ is also a graph property that corresponds to the intersection of $\Phi$ and $\alleven$.
    Then every $k$-vertex graph $H \in \Psi$ has an even number of edges.
    So
    \[\altenum{\Psi}(K_k) = \sum_{S \subseteq E(K_k)} (-1)^{|S|} \cdot \Psi(K_k[S]) = \sum_{S \subseteq E(K_k)} \Psi(K_k[S])\]
    denotes the number of $k$-vertex graphs in $\Psi$.
    Since there is at least one such graph, we conclude that $\altenum{\Psi}(K_k) \neq 0$.

    By Lemma \ref{lem:alternating-enumerator-intersection} there are sets $S,S_{\sf even} \subseteq E(K_k)$ such that $S \cup S_{\sf even} = E(K_k)$ and $\altenum{\Phi}(K_k[S]) \neq 0$ and $\altalleven(K_k[S_{\sf even}]) \neq 0$.
    We set $H \coloneqq K_k[S]$.
    Since $\altalleven(K_k[S_{\sf even}]) \neq 0$, we conclude that $K_k[S_{\sf even}]$ is bipartite by Theorem \ref{thm:alternating-enumerator-alleven}.
    Let $V,W$ denote a bipartition of $K_k[S_{\sf even}]$.
    Then $\binom{V}{2} \subseteq S$ and $\binom{W}{2} \subseteq S$ since $S \cup S_{\sf even} = E(K_k)$.
    In particular, $H$ contains a clique of size at least $\frac{k}{2}$.
\end{proof}

As the second example, we consider invariants $\Phi$ whose intersection with $\alleven$ is non-empty and contains only few graphs.
A similar approach gives the following variant of Theorem \ref{thm:alternating-enumerator-small}.

\begin{theorem}
    \label{thm:alternating-enumerator-alleven-small}
    Let $k \geq 1$, $0 \leq \ell \leq \binom{k}{2} - \left(\frac{k}{2}\right)^2 $ and let $\Phi$ be a $k$-vertex graph invariant such that
    \[1 \leq |\supp(\Phi \cdot \alleven)| \leq 2^{\binom{k}{2} - \left(\frac{k}{2}\right)^2 - \ell}.\]
    Then there is a $k$-vertex graph $H$ such that $\altenum{\Phi}(H) \neq 0$ and $|E(H)| \geq \ell$.
\end{theorem}

\begin{proof}
    Let $\Psi \coloneqq \Phi \cdot \alleven$.
    By Theorem~\ref{thm:alternating-enumerator-small}, there is a $k$-vertex graph $F$ such that $\altenum{\Psi}(F) \neq 0$ and $|E(F)| \geq \left(\frac{k}{2}\right)^2 + \ell$.

    By Lemma~\ref{lem:alternating-enumerator-intersection} there are sets $S,S_{\sf even} \subseteq E(F)$ such that $S \cup S_{\sf even} = E(F)$ and $\altenum{\Phi}(F[S]) \neq 0$ and $\altalleven(F[S_{\sf even}]) \neq 0$.
    We set $H \coloneqq F[S]$.
    Since $\altalleven(F[S_{\sf even}]) \neq 0$, we conclude that $F[S_{\sf even}]$ is bipartite by Theorem \ref{thm:alternating-enumerator-alleven}.
    In particular, $|S_{\sf even}| \leq \left(\frac{k}{2}\right)^2$.
    Because $S \cup S_{\sf even} = E(F)$ and $|E(F)| \geq \left(\frac{k}{2}\right)^2 + \ell$, we get $|E(H)| = |S| \geq \ell$.
\end{proof}

\begin{remark}
    Similar to Theorem~\ref{thm:alternating-enumerator-small}, an earlier version of the paper showed a slightly weaker form of Theorem~\ref{thm:alternating-enumerator-alleven-small} via the uncertainty principle (see Lemma~\ref{lem:uncertainty}).
    Indeed, if $\supp(\Phi)$ is small, then the Fourier representation of $\Phi$ has large support, which implies it contains a graph that is far away from being bipartite.
\end{remark}

By combining Theorem \ref{thm:alternating-enumerator-alleven-small} with Theorems \ref{thm:hardness-w}\ref{item:hardness-w-1} and \ref{thm:hardness-eth}\ref{item:hardness-eth-1}, we obtain the following.

\begin{corollary}[Theorem \ref{thm:intro-hardness-alleven-small} restated]
    \label{cor:hardness-alleven-small}
    For every $0 < \varepsilon < 1$ there are $N_0,\delta > 0$ such that the following holds.
    Let $k \geq N_0$ and let $\Phi$ be a $k$-vertex graph invariant such that
    \[1 \leq |\supp(\Phi \cdot \alleven)| \leq (\sqrt{2} - \varepsilon)^{\binom{k}{2}}.\]
    Then no algorithm solves $\sindsub{\Phi}$ in time $O(n^{\delta \cdot k})$ unless ETH fails.
    
    Moreover, for every computable graph invariant $\Phi$ such that
    \[1 \leq |\supp_k(\Phi \cdot \alleven)| \leq (\sqrt{2} - \varepsilon)^{\binom{k}{2}}\]
    holds for infinitely many $k$, the problem $\pindsub{\Phi}$ is $\sharpwone$-hard.
\end{corollary}

\section{Weisfeiler-Leman Dimension}
\label{sec:wl}

Besides hardness results for induced subgraph counts, our techniques also allow us to bound or even precisely determine the Weisfeiler-Leman dimension of such graph invariants.
The $k$-dimensional Weisfeiler-Leman algorithm ($k$-WL) is a standard heuristic for testing isomorphism of graphs, but it also has surprising connections to various other areas of computer science, including finite model theory and machine learning~\cite{Grohe17,GroheN21,Kiefer20,MorrisLMRKGFB22}.
We omit a detailed description of the WL algorithm here and refer the interested reader to~\cite{Grohe17,GroheN21,Kiefer20}.
Indeed, the following characterization shown by Dvor{\'{a}}k \cite{Dvorak10} is sufficient for our purposes.

A homomorphism from a graph $F$ to a graph $G$ is a mapping $\varphi\colon V(F) \to V(G)$ such that $\varphi(v)\varphi(w) \in E(G)$ for every edge $vw \in E(F)$.
We write $\numhom{F}{G}$ to denote the number of homomorphisms from $F$ to $G$.
Then two graphs $G,G'$ are \emph{distinguished} by $k$-WL if $\numhom{F}{G} \neq  \numhom{F}{G'}$ for some graph $F$ with $\tw(F) \leq k$. 

The \emph{Weisfeiler-Leman dimension} of a graph invariant $f$ is the minimal $k \geq 1$ such that $k$-WL distinguishes between all graphs $G$ and $G'$ for which $f(G) \neq f(G')$.
For the graph invariants considered in this paper, such a number $k$ indeed exists.
Due to the connections of WL to various other areas of computer science, the WL dimension has become a widely studied measure (see, e.g.,~\cite{ArvindFKV20,ArvindFKV22,Furer17,GobelGR24,LanzingerB24,Neuen24}).
Our results yield bounds on the WL dimension of $f(\star) = \numindsubstar{\slice{\Phi}{k}}$ for various graph invariants $\Phi$.

More precisely, combining results from \cite{CurticapeanDM17} and \cite{LanzingerB24,Neuen24}, we establish a connection between the WL dimension and the alternating enumerator of certain graphs with respect to $\Phi$.
We start by covering the necessary tools from \cite{CurticapeanDM17} and \cite{LanzingerB24,Neuen24}.

\begin{definition}
    A graph invariant $\Psi$ is a \emph{graph motif parameter} if there is a finite collection of graphs $\CL$ and coefficients $\alpha_F \in \RR$ for all $F \in \CL$ such that
    \[\Psi(G) = \sum_{F \in \CL} \alpha_F \cdot \numhom{F}{G}\]
    for all graphs $G$.
\end{definition}

The WL dimension of graph motif parameters has been determined in \cite{LanzingerB24,Neuen24}.

\begin{theorem}[\cite{LanzingerB24,Neuen24}]
    \label{thm:wl-dimension-graph-motif}
    Let $\Psi$ be a \emph{graph motif parameter} and suppose there are a finite collection of pairwise non-isomorphic graphs $\CL$ and coefficients $\alpha_F \in \RR \setminus \{0\}$ for all $F \in \CL$ such that
    \[\Psi(G) = \sum_{F \in \CL} \alpha_F \cdot \numhom{F}{G}\]
    for all graphs $G$.
    Then the WL dimension of $\Psi$ is equal to $\max_{F \in \CL}\tw(F)$.
\end{theorem}

Also, it has been shown in \cite{CurticapeanDM17} that, for every graph invariant $\Phi$ on $k$-vertex graphs, the function $\numindsubstar{\Phi}$ is a graph motif parameter.
In fact, \cite{CurticapeanDM17} even provides explicit formulas to determine the set $\CL$ and the coefficients $\alpha_F$ for all graphs $F \in \CL$.
For the purpose of this paper, the following weaker version is sufficient.

For a graph $H$ consider a partition $\rho$ of $V(H)$ such that each block $B \in \rho$ is an independent set.
Then we define $H/\rho$ as the graph obtained from $H$ by contracting each block to a single vertex.
Formally, we set $V(H/\rho) = \rho$ and $E(H/\rho) = \{B_1B_2 \mid \exists v_1 \in B_1, v_2 \in B_2\colon v_1v_2 \in E(H)\}$.
We say a graph $F$ is a \emph{quotient} of $H$, denoted by $F \preceq H$, if $F$ is isomorphic to $H/\rho$ for some partition $\rho$ of $V(H)$ such that each block $B \in \rho$ is an independent set.

\begin{lemma}[\cite{CurticapeanDM17}]
    \label{lem:ind-to-hom-basis}
    Let $\Phi$ be a $k$-vertex graph invariant.
    Then there is a finite collection of pairwise non-isomorphic graphs $\CL$ and coefficients $\alpha_F \in \RR \setminus \{0\}$ such that
    \begin{enumerate}[label=(\Roman*)]
        \item $\numindsub{\Phi}{G} = \sum_{F \in \CL} \alpha_F \cdot \numhom{F}{G}$ for all graphs $G$,
        \item\label{item:ind-to-hom-basis-2} $F \in \CL$ for all graphs such that $\altenum{\Phi}(F) \neq 0$, and
        \item\label{item:ind-to-hom-basis-3} if $F \in \CL$, then there is some $k$-vertex graph $H$ such that $F \preceq H$ and $\altenum{\Phi}(H) \neq 0$.
    \end{enumerate}
\end{lemma}

The next lemma is an immediate consequence of Theorem \ref{thm:wl-dimension-graph-motif} and Lemma \ref{lem:ind-to-hom-basis}.

\begin{lemma}
    \label{lem:wl-lower-bound-from-alternating-enumerator}
    Let $\Phi$ be a $k$-vertex graph invariant.
    Suppose there is a $k$-vertex graph $H$ with $\altenum{\Phi}(H) \neq 0$ and $\tw(H) \geq \ell$.
    Then the WL dimension of $\numindsubstar{\Phi}$ is at least $\ell$.
\end{lemma}

We can combine Lemma \ref{lem:wl-lower-bound-from-alternating-enumerator} with Theorems \ref{thm:alternating-enumerator-small}, \ref{thm:alternating-enumerator-monotone}, \ref{thm:alternating-enumerator-symmetric}, \ref{thm:alternating-enumerator-hw}, \ref{thm:alternating-enumerator-alleven}, \ref{thm:alternating-enumerator-alleven-bipartite} and \ref{thm:alternating-enumerator-alleven-small} as well as Corollary \ref{cor:alternating-enumerator-monotone-small-image} to obtain bounds on the WL dimension of certain graph invariants.
Let us only explicitly mention the following consequence, which follows from Theorem \ref{thm:alternating-enumerator-monotone} by setting $p = 2$ and gives a lower bound that differs from the trivial upper bound of $k-1$ only by the constant factor $4$.

\begin{corollary}
    Let $k \geq 1$ and suppose $\Phi$ is a $k$-vertex graph property that is monotone and not $k$-trivial.
    Then the WL dimension of $\numindsubstar{\Phi}$ is strictly greater than $\frac{k}{4}$.
\end{corollary}

For $\alleven$, we can even determine the WL dimension of $\numindsubstar{\slice{\alleven}{k}}$ exactly.

\begin{theorem}
    \label{thm:wl-alleven}
    For every $k \geq 2$ the WL dimension of $\numindsubstar{\slice{\alleven}{k}}$ is equal to $\left\lfloor\frac{k}{2}\right\rfloor$.
\end{theorem}

The proof relies on the following simple lemma.

\begin{lemma}
    \label{lem:vertex-cover-quotient}
    Let $H$ be a graph and $F$ a quotient of $H$.
    Then $\tw(F) \leq \vc(F) \leq \vc(H)$.
\end{lemma}

\begin{proof}
    The first inequality is folklore.
    For the second inequality, let $\rho$ be a partition of $H$ into independent sets, and let $C$ be a vertex cover of $H$.
    Then $C/\rho \coloneqq \{B \in \rho \mid B \cap C \neq \emptyset\}$ is a vertex cover of $H/\rho$.
    So $\vc(H/\rho) \leq \vc(H)$.
    Since $F$ is a quotient of $H$, the second inequality holds.
\end{proof}

\begin{proof}[Proof of Theorem \ref{thm:wl-alleven}]
    Let $B \coloneqq K_{\left\lfloor k/2 \right\rfloor,\left\lceil k/2 \right\rceil}$.
    We have $\tw(B) = \left\lfloor\frac{k}{2}\right\rfloor$ and
    by Theorem \ref{thm:alternating-enumerator-alleven}, we have $\altalleven(B) \neq 0$.
    So the WL dimension of $\numindsubstar{\slice{\alleven}{k}}$ is at least $\left\lfloor\frac{k}{2}\right\rfloor$ by Lemma \ref{lem:wl-lower-bound-from-alternating-enumerator}.

    For the other direction, let $H$ be a graph such that $\altalleven(H) \neq 0$ and let $F$ be a quotient of $H$.
    Then $H$ is bipartite by Theorem \ref{thm:alternating-enumerator-alleven} which implies that $\tw(F) \leq \vc(H) \leq \left\lfloor\frac{k}{2}\right\rfloor$ by Lemma \ref{lem:vertex-cover-quotient}.
    This implies that the WL dimension of $\numindsubstar{\slice{\alleven}{k}}$ is at most $\left\lfloor\frac{k}{2}\right\rfloor$ by Theorem \ref{thm:wl-dimension-graph-motif} and Lemma \ref{lem:ind-to-hom-basis}\ref{item:ind-to-hom-basis-3}.
\end{proof}

\section{Conclusion}

We show that relatively straightforward algebraic techniques can yield strong complexity results for the problem of counting induced subgraphs satisfying a fixed property $\Phi$.
Our mostly self-contained paper supersedes significant results for this problem obtained over the last decade via different techniques.
Nevertheless, several open problems remain.

Most importantly, we are still lacking a complete understanding for which properties $\Phi$ the parameterized problem $\pindsub{\Phi}$ is $\sharpwone$-hard.
Recently, together with D{\"{o}}ring \cite{CurticapeanDN25}, we identified several non-meager graph properties $\Psi$ for which $\pindsub{\Phi}$ is polynomial-time solvable, refuting an earlier conjecture on the complexity of $\pindsub{\Phi}$ \cite{DorflerRSW22,FockeR24,RothSW24}.
This led to the following updated conjecture:

\begin{conjecture}[{\cite[Conjecture 4.2]{CurticapeanDN25}}]
    Let $\Phi$ be a computable graph property and let $(H_k)_{k \geq 1}$ be a sequence of graphs of unbounded vertex cover number such that $\altenum{\Phi}(H_k) \neq 0$ for all $k \geq 1$.
    Then $\pindsub{\Phi}$ is $\sharpwone$-hard.
\end{conjecture}

A version of this conjecture for graph \emph{invariants} rather than properties is known to fail.

Moreover, it would be interesting to obtain tighter lower bounds for the Weisfeiler-Leman dimension of $\numindsubstar{\Phi}$ for $k$-vertex graph invariants $\Phi$.
Such bounds translate into tighter inexpressibility results in logics and computational lower bounds for symmetric circuits.

\bibliographystyle{plainurl}
\bibliography{indsub-fourier}

\appendix

\section{Omitted Proofs}
\label{app:omitted-proofs}

\subsection{Vertex-Colored Graphs}

In the appendix, we require formal definitions regarding colored graphs that are not relevant for the main part.
A \emph{vertex-colored graph} with colors $C$ is an undirected graph $G=(V,E,c)$ that is given together with a \emph{coloring} $c\colon V(G) \to C$.
Write $G^{\circ}=(V,E)$ for the uncolored graph underlying $G$.
We call $G$ \emph{colorful} if $c$ is a bijection, and we call $G$ \emph{canonically colored} if $C = V(G)$ and $c$ is the identity function.

For $i\in C$, we write $V_{i}(G)$ for the vertices of color $i$.
For $i,j\in C$, we write $E_{ij}(G)$ for the edges in $G$ with one endpoint of color $i$ and another of color $j$.
For $X \subseteq C$ and $Y \subseteq \binom{C}{2}$, let $G_{\setminus X,Y}$ be the colored graph obtained from $G$ by deleting all vertices with colors from $X$ and all edges whose endpoints have a color pair from $Y$, i.e,
\[G_{\setminus X,Y}=\left(V\setminus\bigcup_{i\in X}V_{i},\ E\setminus\bigcup_{ij\in Y}E_{ij}\right).\]
Given vertex-colored graphs $H$ and $G$, we write $\numsub{H}{G}$ for the number of subgraphs $F$ of $G$ that are isomorphic to $H$ by an isomorphism that respects colors; for colored graphs $F$ and $H$ with colorings $c_{F}$ and $c_{H}$, this is an isomorphism $\pi\colon V(F) \to V(H)$ with $c_{F}(v)=c_{H}(\pi(v))$ for all $v\in V(H)$.
We write $F \cong F'$ to denote that two (colored or uncolored) graphs $F,F'$ are isomorphic.

Given a canonically colored graph $H$, a graph $G$ on the same set of vertex-colors is \emph{$H$-colored} if $E_{ij}(G) = \emptyset$ for all $ij \notin E(H)$.
The following fact is immediate.

\begin{fact}
    \label{fact:colsub-preprocess}
    For vertex-colored graphs $H$ and $G$ with canonically colored $H$, we have 
    \[\numsub{H}{G} = \numsub{H}{G_{\setminus\emptyset,\overline{E(H)}}}.\]
\end{fact}

Thus, when considering $\numsub{H}{G}$ for graphs $H$ and $G$ with canonically colored $H$, we may assume $G$ to be $H$-colored without loss of generality.

\subsection{Isolating Colorful Subgraphs from Linear Combinations}

As in similar works~\cite{Curticapean15,AminiFS12}, we use the inclusion-exclusion principle for reductions between subgraph counting problems. More specifically, inclusion-exclusion allows us to count subgraphs not contained in any \emph{bad} sets $B_1,\ldots,B_s$ (e.g., avoiding a color they should contain) by determining the cardinalities of intersections of bad sets. We use a version with additive weights that appears, e.g., in elementary probability theory.

\begin{lemma}[see, e.g., {\cite[Lemma~1.3]{MitzenmacherU05}}]
    \label{lem:incl-excl}
    Let $\Omega$ be a finite set with weights $w:\Omega \to \FF$ for a field $\FF$.
    Given a subset $F \subseteq \Omega$, define its weight as $w(F) \coloneqq \sum_{x\in F} w(x)$.
    Let $B_{1},\ldots,B_{s}\subseteq\Omega$.
    For $X\subseteq[s]$, write $B_{X} \coloneqq \bigcap_{i\in X}B_{i}$ with the convention $B_{\emptyset}=\Omega$. 
    Then
    \[ w\left(\Omega\setminus\bigcup_{i=1}^{s}B_{i}\right)
    =
    \sum_{X\subseteq[s]}(-1)^{|X|}\,w(B_{X}).\]
\end{lemma}

Using this, we can compute $\numsub{H}{G}$ for vertex-colored graphs $H$ and $G$, with colorful $H$, provided access to a linear combination of \emph{uncolored} subgraph counts $f(\star)$ that contains $\numsubstar{H}$ with non-zero coefficient. 

\begin{lemma}
    \label{lem:sub-monotonicity}
    Let $k\in\mathbb{N}$ and $\FF$ be a field. 
    For coefficients $\alpha_{1},\ldots,\alpha_{s}\in\FF$
    and pairwise non-isomorphic uncolored graphs $H_{1},\ldots,H_{s}$, each on exactly $k$ vertices, consider the graph invariant
    \begin{equation}
    \label{eq:appendix-graphmotif}
        f(\star) \coloneqq \sum_{i=1}^{s}\alpha_{i} \cdot \numsubstar{H_{i}}.
    \end{equation}
    Let $H$ be a canonically colored graph with $H^{\circ} \cong H_b$ for some $b\in[s]$, and let $G$ be an $H$-colored graph.
    Then we have
    \[
    \alpha_{b}\cdot\numsub{H}{G}=\sum_{\substack{X\subseteq V(H)\\
    Y\subseteq E(H)
    }
    }(-1)^{|X|+|Y|}f(G_{\setminus X,Y}^{\circ}).
    \]
\end{lemma}

Note that this lemma handles two separate issues for us: First, it allows us to evaluate one \emph{individual} subgraph count $\numsub{H}{G}$
with an oracle for evaluating a linear combination of subgraph counts.
Moreover, it allows us to reduce subgraph counts from \emph{colorful} graphs to evaluations of $f$ on uncolored graphs.
As it turns out, these two issues are best handled together; this was recently done similarly for linear combinations of homomorphism counts~\cite{Curticapean24}.

\begin{proof}
    For $i\in [s]$, denote the set of all $H_i$-subgraph copies in $G^\circ$ as
    \[\Omega_i \coloneqq \{F\subseteq G^{\circ}\mid F\cong H_i\}.\]
    These sets are pairwise disjoint, since the graphs $H_1,\ldots,H_s$ are pairwise non-isomorphic.
    Let $\Omega \coloneqq \bigcup_{i=1}^s \Omega_i$ and endow $\Omega$ with a weight function $w:\Omega\to\mathbb R$ by setting $w(F)=\alpha_i$ if $F \in \Omega_i$.
    Then the sum in \eqref{eq:appendix-graphmotif} can be rephrased as $w(\Omega)$. 
    
    Our aim is to isolate the graphs $F\in\Omega$ that are isomorphic to $H^\circ$. 
    We achieve this by using the weighted inclusion-exclusion principle to enforce two conditions on graphs $F \in \Omega$ to be counted in terms of ``bad subsets'' of $\Omega$:

    \begin{itemize}
        \item $F$ should be vertex-colorful under the coloring of $G$. This is equivalent to requiring $V_{i}(F)\neq\emptyset$ for all colors $i\in[k]$,
        since $F$ has exactly $k$ vertices by the prerequisites of the lemma. We capture this by defining, for color $i\in[k]$, the bad set $B_{i}:=\{F\in\Omega\mid V_{i}(F)=\emptyset\}$
        of subgraphs $F$ that miss color $i$.
        \item $F$ should have $E_{ij}(F)\neq\emptyset$ for all $ij\in E(H)$.
        For $ij\in E(H)$, we therefore define the bad set $B_{ij}:=\{F\in\Omega\mid E_{ij}(F)=\emptyset\}$
        violating this condition for $ij$.
    \end{itemize}
    Write $A \subseteq \Omega$ for the set of ``good'' graphs $F \in \Omega$, i.e.,
    \[A = \Omega \setminus \left( \bigcup _{i\in [k]}B_i  \ \cup   \bigcup _{ij\in E(H)}B_{ij} \right).\]
    Given $F \in \Omega$, we write $F^{c}$ for the colored graph obtained by inheriting the colors from $G$.
    We claim that $A^c \coloneqq \{F^c \mid F\in A\}$ is precisely the set of color-respecting $H$-copies in $G$, which implies that $w(A)=\alpha_b \cdot \numsub{H}{G}$ as required. 
    
    To see this claim, observe that the following statements hold for every graph $F^{c}$ with $F \in A$:
    \begin{itemize}
        \item $F^{c}$ is \emph{colorful}: At least one vertex from each color class of $G$ is contained in $F$, because $F\notin B_i$ for all $i \in [k]$, and $F$ has $k$ vertices in total.
        \item Because $G$ is $H$-colored and $F^{c}$ is a colorful subgraph of $G$, it follows that $F^{c} \cong H'$ for a \emph{subgraph} $H' \subseteq H$.
        \item Moreover, since $F\notin B_{ij}$ for all $ij \in E(H)$, we have $E_{ij}(F) \neq \emptyset$ for $ij\in E(H)$. It follows that $F^c$ is isomorphic to a \emph{supergraph} of $H$.
    \end{itemize}
    Since $F^c$ is colorful and isomorphic to both a subgraph and a supergraph of $H$, it follows that $F^c$ is actually isomorphic to $H$.
    This shows that $A$ contains precisely the colorful $H$-copies in $G$ and proves the claim by Lemma~\ref{lem:incl-excl}.
\end{proof}

\subsection{Complexity Lower Bounds by Isolating Subgraphs}

We are ready to use Lemma~\ref{lem:sub-monotonicity} to obtain complexity lower bounds.
First, the following useful fact allows us to restrict alternating enumerators to slices.

\begin{fact}
    \label{fact:slice-restriction}
    For every graph invariant $\Phi$ and $k \geq 1$, it holds that $\altenum{(\slice{\Phi}{k})}(H) = \altenum{\Phi}(H)$ for all $k$-vertex graphs $H$.
\end{fact}

\begin{proof}
    This follows directly from Definition \ref{def:alt-enum}, since $\altenum{\Phi}(H)$ depends only on the values of $\Phi$ on $k$-vertex graphs.
\end{proof}

We base our hardness results on the following parameterized problem:
For a fixed class of uncolored graphs $\mathcal H$, the problem $\pcolsub{\mathcal H}$ asks, on input an uncolored graph $H \in \mathcal H$ and a vertex-colored graph $G$, to determine the number $\numsub{\can{H}}{G}$, where $\can{H}$ denotes the canonically colored version of $H$.
The problem is parameterized by $|V(H)|$.
Under different names, it has featured in a similar role as a reduction source throughout parameterized counting complexity, e.g., in~\cite{Curticapean15,CurticapeanDM17,CurticapeanM14,DoringMW24,RothSW24}, and it is known to be hard if $\mathcal H$ has unbounded treewidth and satisfies very mild constructibility conditions, e.g., recursive enumerability.

\begin{theorem}
    \label{thm:grid-minor}
    If $\mathcal{H}=(H_{t})_{t\geq1}$ is a recursively enumerable sequence of graphs of unbounded treewidth,
    then $\pcolsub{\mathcal H}$ is $\sharpwone$-hard and $\pmodcolsub{p}{\mathcal H}$ is $\modwone{p}$-hard for every prime $p$.
\end{theorem}

\begin{proof}
    For integer $\ell \geq 1$, let $\boxplus_{\ell}$ denote the uncolored $\ell \times \ell$ square grid, and let $\mathcal H_\boxplus =\{\boxplus_{\ell} \mid \ell \geq 1\}$ denote the class of grids.
    The problem $\pcolsub{\mathcal H_\boxplus}$ is \sharpwone-hard by a parsimonious parameterized reduction from counting $k$-cliques (cf.~\cite[Lemma~5.7]{Curticapean15}), so $\pcolsub{\mathcal H_\boxplus}$ and $\pmodcolsub{p}{\mathcal H_\boxplus}$ for primes $p$ are $\sharpwone$-hard and $\modwone{p}$-hard, respectively.
     
    The excluded grid theorem~\cite{RobertsonS86,ChekuriC16} implies that, for every $\ell\in\mathbb{N}$, there exists an index $t(\ell)$ such that $\boxplus_{\ell}$ is a minor of $H_{t(\ell)}$. 
    Consider two recursively enumerable graph classes $\mathcal{F}$ and $\mathcal{H}$ such that, for every graph $F\in \mathcal F$, there is a graph $H \in \mathcal H$ containing $F$ as a minor:
    A parsimonious parameterized reduction from $\pcolsub{\mathcal F}$ to $\pcolsub{\mathcal H}$ is known in this setting~\cite[Lemma~5.8]{Curticapean15}.
    Taking $\mathcal F = \mathcal H_\boxplus$ and $\mathcal H$ as above, it follows that $\pcolsub{\mathcal H}$ is $\sharpwone$-hard and $\pmodcolsub{p}{\mathcal H}$ is $\modwone{p}$-hard for every prime $p$.
\end{proof}

\begin{proof}[Proof of Theorem~\ref{thm:hardness-w}]
    We reduce $\pcolsub{\mathcal H}$ to $\pindsub{\Phi}$ via Theorem~\ref{thm:grid-minor} and Lemma~\ref{lem:sub-monotonicity}, which carry over to the modular setting.
    Towards this, let $\Phi$ be a computable graph invariant. We describe a procedure that constructs a graph sequence $\mathcal{H}=(H_{t})_{t\geq1}$ with $\altenum{\Phi}(H_{t})\neq0$ and $\tw(H_t)\geq t$ for all $t \in \NN$:
    Maintain a variable $T \in \NN$ for the largest treewidth seen so far, and for increasing $k\in \NN$, enumerate all graphs on vertex set $[k]$ in lexicographic order.
    Whenever this process reaches a graph $F$ with $\altenum{\Phi}(F)\neq0$ and treewidth $t>T$ (where both conditions on $F$ are computable), set $H_i \coloneqq F$ for all $i$ with $T <i \leq t$, and update $T\coloneqq t$.
    By the assumed sequence of unbounded treewidth in the theorem statement, this sequence is infinite, and by construction, $\mathcal{H}$ is recursively enumerable.
    
    Theorem~\ref{thm:grid-minor} shows that $\pcolsub{\mathcal H}$ is $\sharpwone$-hard. This problem in turn reduces to $\pindsub{\Phi}$ via Lemma~\ref{lem:sub-monotonicity}:
    Consider an instance $(H,G)$ for $\pcolsub{\mathcal H}$ and let $k \coloneqq |V(H)|$ and $n \coloneqq |V(G)|$.
    We wish to determine $\numsub{\can{H}}{G}$. 
    Using Fact~\ref{fact:colsub-preprocess}, we may assume that $G = G_{\setminus\emptyset,\overline{E(H)}}$.
    By Corollary~\ref{cor:phi=sub} we have
    \[\numindsub{\slice{\Phi}{k}}{\star} = \sum_F \altenum{(\slice{\Phi}{k})}(F) \cdot \numsub{F}{\star} =: f(\star),\]
    where $F$ ranges over all $k$-vertex graphs.
    Invoking Lemma~\ref{lem:sub-monotonicity}, we obtain that
    \begin{equation}
        \label{eq:using-submon-in-reduction}
        \altenum{(\slice{\Phi}{k})}(H) \cdot \numsub{H}{G} = \sum_{\substack{X\subseteq V(H)\\Y\subseteq E(H)}}(-1)^{|X|+|Y|}f(G_{\setminus X,Y}^{\circ}).
    \end{equation}
    Since $\altenum{(\slice{\Phi}{k})}(H) = \altenum{\Phi}(H)$ by Fact \ref{fact:slice-restriction} and $\altenum{\Phi}(H) \neq 0$ by assumption, we obtain a parameterized Turing reduction from $\pcolsub{\mathcal H}$ to $\pindsub{\Phi}$ by evaluating the right-hand side of \eqref{eq:using-submon-in-reduction} and dividing by $\altenum{\Phi}(H) \neq 0$.
    The division is admissible over every field $\FF$.
    
    Note that all values $f(G_{\setminus X,Y}^{\circ})$ can be obtained by the oracle calls $\numindsub{\slice{\Phi}{k}}{G^\circ_{\setminus X,Y}}$, without parameter increase, in overall time $2^{|V(H)|+|E(H)|} \cdot n^{O(1)}$.
    The value $\altenum{\Phi}(H)$ can be computed by brute-force evaluating $\Phi$ on $2^{O(k^2)}$ many $k$-vertex graphs.
    This yields the first statement of the theorem.

    For the second statement, note that $\ZZ_p$ for prime $p$ is a field, so the right-hand side of \eqref{eq:using-submon-in-reduction} is well-defined and the division is admissible.
    Thus, the above approach also yields a parameterized reduction from $\pmodcolsub{p}{\mathcal H}$ to $\pmodindsub{p}{\Phi}$.
\end{proof}

Next, we give a proof for Theorem~\ref{thm:hardness-eth}.
For a \emph{fixed} graph $H$, the problem $\colsub{H}$ asks, given as input a vertex-colored graph $G$, to determine the number $\numsub{\can{H}}{G}$, where $\can{H}$ denotes the canonically colored version of $H$. Note that only $G$ is part of the input, while $H$ is fixed.
As usual, we use $\modcolsub{p}{H}$ to denote the modular counting version.
Also, we write $\deccolsub{H}$ to denote the the decision version, i.e., $\deccolsub{H}$ asks to decide whether $\numsub{\can{H}}{G} > 0$.

The proof of Theorem~\ref{thm:hardness-eth} builds on Theorem~\ref{fact:density-makes-hard}, which gives lower bounds for $\deccolsub{H}$ for all graphs $H$ with large edge-density, cf.~\cite{CurticapeanDNW25}.
For Theorem~\ref{thm:hardness-eth}\ref{item:hardness-eth-3}, we need the lower bound of Theorem~\ref{fact:density-makes-hard} even for the decision version rather than the counting version.

\begin{theorem}[{\cite[Theorem 1.3]{CurticapeanDNW25}}]
    \label{fact:density-makes-hard-decision}
    Assuming ETH, there is a constant $\beta > 0$ such that, for every graph $H$ with average degree $d$, no algorithm solves $\deccolsub{H}$ in time $O(n^{\alpha \cdot d})$.
\end{theorem}

\begin{proof}[Proof of Theorem~\ref{thm:hardness-eth}]
    We first prove Item~\ref{item:hardness-eth-1}.
    Let $\Phi$ be a $k$-vertex graph invariant and suppose $H$ is a graph with $\altenum{\Phi}(H) \neq 0$ and $E(H) \geq k \cdot \ell \geq N_0$.
    Without loss of generality assume $V(H) = [k]$.
    We give an algorithm for $\colsub{H}$ that uses an algorithm for $\sindsub{\Phi}$ as a subroutine.

    Let $G$ be the input graph for the problem $\colsub{H}$.
    We may assume that $G = G_{\setminus\emptyset,\overline{E(H)}}$ by Fact~\ref{fact:colsub-preprocess}.
    We wish to determine $\numsub{\can{H}}{G}$.
    By Corollary~\ref{cor:phi=sub} we have
    \[\numindsub{\Phi}{\star} = \sum_F \altenum{\Phi}(F) \cdot \numsub{F}{\star} =: f(\star),\]
    where $F$ ranges over all $k$-vertex graphs.
    Invoking Lemma~\ref{lem:sub-monotonicity}, we obtain again that
    \begin{equation}
        \label{eq:using-submon-in-reduction-2}
        \altenum{\Phi}(H) \cdot \numsub{\can{H}}{G} = \sum_{\substack{X\subseteq V(H)\\Y\subseteq E(H)}}(-1)^{|X|+|Y|}f(G_{\setminus X,Y}^{\circ}).
    \end{equation}
    We can therefore compute $\numsub{\can{H}}{G}$ by evaluating the right-hand side of \eqref{eq:using-submon-in-reduction-2} and dividing by $\altenum{\Phi}(H) \neq 0$.
    Note that all relevant values $f(G_{\setminus X,Y}^{\circ})$ can be obtained by the oracle calls $\numindsub{\Phi}{G^\circ_{\setminus X,Y}}$ without parameter increase in overall time $2^{|V(H)|+|E(H)|} \cdot n^{O(1)}$.
    The value $\altenum{\Phi}(H)$ can be computed by brute-force by evaluating $\Phi$ on $2^{O(k^2)}$ many $k$-vertex graphs.
    Hence, an $O(n^{\beta \cdot \ell})$ algorithm for $\sindsub{\Phi}$ gives an $O(n^{c \cdot \beta \cdot \ell})$ for $\colsub{H}$ for some suitable fixed constant $c$.
    Now, Item~\ref{item:hardness-eth-1} follows from Theorem~\ref{fact:density-makes-hard-decision}.

    For Item~\ref{item:hardness-eth-3}, we first observe that, by the same arguments as before, an $O(n^{\beta \cdot \ell})$ algorithm for $\smodindsub{p}{\Phi}$ gives an $O(n^{c \cdot \beta \cdot \ell})$ for $\modcolsub{p}{H}$ for some suitable fixed constant $c$.
    We can design a randomized algorithm for $\deccolsub{H}$ as follows:
    Given an input graph $G$, we pick a subset $X \subseteq V(G)$ uniformly at random, and by \cite[Lemma~2.1]{WilliamsWWY15}, with probability at least $2^{-k}$, we have $\numsub{\can{H}}{G} > 0$ if and only if $\numsub{\can{H}}{G[X]} \neq 0 \bmod p$.
    The latter can be decided using the oracle for $\modcolsub{p}{H}$.
    Now, Item~\ref{item:hardness-eth-3} also follows from Theorem~\ref{fact:density-makes-hard-decision}.
\end{proof}

Finally, we show an algorithmic upper bound omitted from the main part.

\begin{proof}[Proof of Theorem~\ref{thm:algorithm-alternating-enumerator-vc}]
    Let $(G,k)$ be an input for $\pindsub{\Phi}$, where $G$ is a graph and $k \geq 1$ is an integer.
    We wish to compute $\numindsub{\slice{\Phi}{k}}{G}$ for the $k$-th slice $\slice{\Phi}{k}$ of $\Phi$. 
    Corollary~\ref{cor:phi=sub} yields, with $H$ ranging over all unlabelled $k$-vertex graphs,
    \begin{equation}
        \label{eq:numphi-via-sub}
        \numindsub{\slice{\Phi}{k}}{G} = \sum_H \altenum{(\slice{\Phi}{k})}(H) \cdot \numsub{H}{G}.
    \end{equation}
    Every graph $H$ with non-zero coefficient $\altenum{(\slice{\Phi}{k})}(H) = \altenum{\Phi}(H)$ (recall Fact~\ref{fact:slice-restriction} for the equality) satisfies $\vc(H) \leq t(k)$.
    Using Definition~\ref{def:alt-enum}, all coefficients for all $k$-vertex graphs $H$ can be computed in time $2^{O(k^2)}$ by brute-force.
    Thus, it remains to compute $\numsub{H}{G}$ for each of the at most $2^{O(k^2)}$ graphs $H$ with non-zero coefficient $\altenum{\Phi}(H)$; the non-zeroness of the coefficient implies $\vc(H) \leq t(k)$.

    By \cite[Theorem~1.1]{CurticapeanDM17}, the value $\numsub{H}{G}$ on input graphs $H$ and $G$, with $k$ and $n$ vertices respectively, can be computed in time $f(k)\cdot n^{q(H)+1}$, where $q(H) \coloneqq \max_{F \preceq H} \tw(F)$ and $f$ is computable. Here $F \preceq H$ denotes that $F$ is a quotient of $H$, i.e., it can be obtained repeatedly identifying vertices; see also Section~\ref{sec:wl}. 
    Given a graph $H$ with $\vc(H)\leq t(k)$, we have $q(H) \leq t(k)$ by Lemma \ref{lem:vertex-cover-quotient}. 
    Using \cite[Theorem~1.1]{CurticapeanDM17}, it follows that an individual value $\numsub{H}{G}$ for $H$ with $\altenum{(\slice{\Phi}{k})}(H) = \altenum{\Phi}(H) \neq 0$ can be computed in time $f(k)\cdot n^{t(k)+1}$.
    Hence, $\numindsub{\slice{\Phi}{k}}{G}$ can be computed in time $2^{O(k^2)} f(k)\cdot n^{t(k)+1}$ via \eqref{eq:numphi-via-sub}.
\end{proof}

\section{Few Graphs Avoid a Fixed Induced Graph}
\label{sec:F-avoid}

A stronger bound on the asymptotic number of induced $F$-free graphs for fixed $F$ is known~\cite{PromelS92}, but it requires a  significantly more sophisticated proof. For completeness, we give an elementary proof of a weaker upper bound.
In the following, an \emph{induced $F$-copy} in $G$ is an induced subgraph of $G$ isomorphic to $F$.

\begin{theorem}
    \label{thm:Fcopies}
    For every fixed graph $F$, there is a constant $\varepsilon > 0$ such that, for $k$ large enough, at most $(2-\varepsilon)^{\binom{k}{2}}$ labeled graphs on vertex set $[k]$ do not contain an induced $F$-copy.
\end{theorem}

\begin{proof}
    We show that, for every fixed $F$, there is a constant $c>0$ such that, for $k$ large enough, a random graph $H\sim\mathcal{G}(k,1/2)$ contains an induced $F$-copy with probability at least $1-\exp(-ck^{2})$.
    This implies the theorem, since then at most $\exp(-ck^{2}) \cdot  2^{\binom{k}{2}} \leq (2-\varepsilon)^{\binom{k}{2}}$ graphs for constant $\varepsilon > 0$ have \emph{no} induced $F$-copy.

    In our inductive proof, we show a stronger statement: Let $A_{k,F,t}$ be the event that a random graph $H\sim\mathcal{G}(k,1/2)$ contains
    at least $t$ vertex-disjoint induced $F$-copies, and let
    \[p(k,b,t) \coloneqq \min_{b\text{-vertex graph }F}\mathrm{Pr}(A_{k,F,t}).\]
    Note that $p(k,1,t)=1$ for all $t\leq k$. In the main part of the proof, we show:

    \begin{claim}
        \label{claim:Fcopies-induction}
        There are constants $\beta_{1}\geq\beta_{2}\geq\ldots$ such that, for even $b$ and large enough even $k$,
        \begin{equation}
            \label{eq:Fcopy-recurse}
            p(k,b,t)\geq p(k/2,b/2,2t)^{2}\cdot\left[1-\exp(-\beta_{b}\cdot t^{2})\right].
        \end{equation}
    \end{claim}

    Before proving this, let us first show how the theorem follows from \eqref{eq:Fcopy-recurse} by induction:
    First, we can assume by padding that $k$ and $b$ are powers of two.
    Then we have
    \begin{align*}
        p(k,b,1) & \geq p(k,b,k/b^{2})\\
        & \geq p(k/2,b/2,2k/b^{2})^{2}\cdot\left[1-\exp(-\beta_{b}\cdot k^{2}/b^4)\right]\geq\ldots \\ & \geq\prod_{s=0}^{\log_{2}b-1}\left[1-\exp(-\beta_{b/2^{s}} \cdot 4^s \cdot k^{2}/b^4 )\right]^{2^{s}} \geq 1-\exp(-0.99\cdot \beta_{b}\cdot k^{2}/b^4).
    \end{align*}
    Here, the first inequality is trivial, as $k/b^{2}$ disjoint $F$-copies in particular give \emph{one} $F$-copy.
    The next inequalities are inductive applications of (\ref{eq:Fcopy-recurse}), terminating at $p(k/b,1,k/b)=1$.
    The last inequality follows from standard bounds and the fact that $\beta_{b}\leq\beta_{r}$ for all $r\leq b$, so the factor $1-\exp(-\beta_b\cdot k^2/b^4)$ grows slowest in the overall product over $s$.

    In the remainder of the proof, we show Claim~\ref{claim:Fcopies-induction}.
    Let $H\sim\mathcal{G}(k,1/2)$ and let $F$ be a $b$-vertex graph minimizing $\mathrm{Pr}(A_{k,F,t})$.
    Partition $F$ equitably into $F_{L}$ and $F_{R}$, and partition $H$ equitably into $H_{L}$ and $H_{R}$, by choosing the first half of vertices in the respective graphs.
    With probability at least $p(k/2,b/2,2t)^{2}$, there are $2t$ vertex-disjoint induced $F_{L}$-copies $F_{L}^{1},\ldots,F_{L}^{2t}$ in $H_{L}$, and analogous copies $F_{R}^{1},\ldots,F_{R}^{2t}$ in $H_{R}$.
    Note that these two events are independent.

    Construct an auxiliary bipartite graph $X$ on vertices $\ell_{1},\ldots,\ell_{2t}$ and $r_{1},\ldots,r_{2t}$ corresponding to the copies of $F_{L}$ and $F_{R}$ in $H$.
    The edge $\ell_{i}r_{j}$ is present in $X$ exactly if $V(F_{L}^{i})\cup V(F_{R}^{j})$ induces an $F$-copy in $H$.

    \begin{claim}
        The distribution of $X$ follows a \emph{bipartite} random graph model with $\alpha k$ vertices and independent edge probability $q=q(b)>0$.
    \end{claim}
    \begin{claimproof}
        Indeed, $\ell_{i}r_{j}\in E(X)$ occurs with constant probability $q$, depending only on $b$. Moreover, these events are mutually independent:
        Whether $V(F_{L}^{i})\cup V(F_{R}^{j})$ induces an $F$-copy in $H$ depends only on the ``cross-edges'' $C_{i,j}=V(F_{L}^{i})\times V(F_{R}^{j})$, since $F_{L}^{i}$ and $F_{R}^{j}$ are already proper copies of $F_{L}$ and $F_{R}$.
        The cross-edges $C_{i,j}$ and $C_{i',j'}$ in $H$ are disjoint for $(i,j)\neq(i',j')$, and any events about disjoint edge-sets in $H$ are independent.
    \end{claimproof}

    Next, observe that every $t$-edge matching $M$ in $X$ yields $t$ vertex-disjoint induced $F$-copies in $H$.
    A simple union bound allows us to show that such a matching is very likely to occur.
    \begin{claim}
        There is a constant $\beta_b >0$ such that, for $t$ large enough, the random graph $X$ (i.e., a bipartite random graph on $2t+2t$ vertices with constant edge probability) contains a $t$-edge matching with probability at least $1-\exp(-\beta_b \cdot t^{2})$.
    \end{claim}
    \begin{claimproof}
        The absence of such a matching in $X$ would imply the existence of an independent set in $X$ with $t$ left and $t$ right vertices.
        There are at most $\binom{2t}{t}^{2}\leq2^{4t}$ many candidate sets, and each has probability $(1-q)^{t^{2}}$.
        This upper-bounds the absence probability by $2^{4t}(1-q)^{t^{2}}\leq\exp(-\beta_{b}t^{2})$ for a constant $\beta_{b}>0$ depending on $b$ and large enough $t$.
    \end{claimproof}
    Overall, We have shown Claim~\ref{claim:Fcopies-induction}. This concludes the proof of the theorem.
\end{proof} 

\end{document}